\documentclass[a4paper,11pt]{amsart}
\usepackage{amssymb}
\usepackage{amsmath}

\usepackage[utf8]{inputenc}
\usepackage{amsthm}
\usepackage{amsfonts}
\usepackage{enumerate}
\usepackage{mathrsfs}
\usepackage{enumitem}
\usepackage[ruled,vlined]{algorithm2e}
\usepackage{hyperref}
\usepackage{pdfpages}
\usepackage{float}

\usepackage{pgfplots}
\pgfplotsset{compat=1.15}
\usetikzlibrary{arrows}
\usepackage{caption}
\usepackage[labelformat=simple]{subcaption}

\usepackage{xcolor}

\usepackage{young}
\usepackage[vcentermath]{youngtab}

\hypersetup{
  colorlinks   = true, 
  urlcolor     = blue, 
  linkcolor    = blue, 
  citecolor   = red 
}

\usepackage[a4paper]{geometry}
\usepackage{titletoc}
\usepackage{fancyhdr}
\newtheorem{thm}{Theorem}
\newtheorem{prop}[thm]{Proposition}
\newtheorem{lem}[thm]{Lemma}
\newtheorem{coro}[thm]{Corollary} 

\newtheorem*{question*}{Question}
\theoremstyle{definition}
\newtheorem{defi}[thm]{Definition}

\theoremstyle{remark} 
\newtheorem{rem}[thm]{Remark}
\newtheorem*{rem*}{Remark}
\theoremstyle{definition} 
\newtheorem{ex}[thm]{Example}

\newtheorem*{nota*}{Notation}

\newcommand{\K}{{\mathbf K}}
\newcommand{\Q}{\overline{\mathbb Q}}
\newcommand{\C}{\mathbb{C}}
\newcommand{\Z}{\mathbb{Z}}

\newcommand{\puip}{\phi_p}

\newcommand{\I}{\mathcal{I}}

\newcommand{\R}{\mathbb R}

\newcommand{\val}{\operatorname{val}_{z}}
\DeclareMathOperator{\ev}{ev}
\newcommand{\cld}{\operatorname{cld}_{z}}
\newcommand{\supp}{\operatorname{supp}_{z}}

\newcommand{\Hahn}{\mathscr{H}}
\newcommand{\Rspe}{\mathscr{R}}
\newcommand{\Puis}{\mathscr{P}}
\newcommand{\logm}{{\ell}}

\newcommand{\Px}{\mathcal{P}}

\begin{document}
\title[]{Regular Singular Mahler equations and Newton polygons}

\author{Colin Faverjon}
\address{Univ Lyon, Universit\'e Claude Bernard Lyon 1, CNRS UMR 5208, Institut Camille Jordan, F-69622 Villeurbanne, France}
\email{colin.faverjon@math.cnrs.fr}
\author{Marina Poulet}
\address{Institut Fourier, Universit\'e Grenoble Alpes CS 40700, 38058 Grenoble cedex 9, France}
\email{marina.poulet@univ-grenoble-alpes.fr}

\date{\today}

\begin{abstract} 
Though Mahler equations have been introduced nearly one century ago, the study of their solutions is still a fruitful topic for research.  
In particular, the Galois theory of Mahler equations has been the subject of many recent papers. Nevertheless,  long is the way to a complete understanding of relations between solutions of Mahler equations. One step along this way is the study of singularities. Mahler equations with a regular singularity at $0$ have rather ``nice'' solutions: they can be expressed with the help of Puiseux series and solutions of equations with constant coefficients. In a previous paper, the authors described an algorithm to determine whether an equation is regular singular at $0$ or not. Exploiting information from the Frobenius method and Newton polygons, we improve this algorithm by significantly reducing its complexity, by providing some simple criterion for an equation to be regular singular at $0$, and by extending its scope to equations with Puiseux coefficients.
\end{abstract}
\maketitle

\vspace{-1.5cm}
\setcounter{tocdepth}{1}
\setcounter{secnumdepth}{3}
\titlecontents{section}
[0.5em]
{}
{\contentslabel{1.3em}}
{\hspace*{-2.3em}}
{\titlerule*[1pc]{.}\contentspage}
{\addvspace{2em}\bfseries\large}
\titlecontents{subsection}
[3.8em]
{} 
{\contentslabel{2em}}
{\hspace*{-3.2em}}
{\titlerule*[1pc]{.}\contentspage}
\titlecontents{subsubsection}
[6.1em]
{} 
{\contentslabel{2.4em}}
{\hspace*{-4.1em}}
{\titlerule*[1pc]{.}\contentspage}
\renewcommand{\contentsname}{Contents}

 \pdfbookmark[0]{\contentsname}{Contents}
\tableofcontents

\section{Introduction}

Let $\K$ be a field and $p\geq 2$ be an integer. A $p$-Mahler equation is a functional linear equation of the following form
$$
a_m(z)f(z^{p^m})+ \cdots + a_1(z)f(z^p) + a_0(z)f(z) =0
$$
where $a_0,\ldots,a_m$ lie on some extension of $\K[z]$ and $a_0a_m\neq 0$. When $\K=\Q$ and $a_0,\ldots,a_m\in \Q[z]$, power series solutions of $p$-Mahler equations are known as \textit{$M_p$-functions}, or $M$-functions. These functions has been introduced by Mahler \cite{Ma29, Ma30a, Ma30b} to study the algebraic relations between their values at non-zero algebraic points. This subject has been the topic of many papers (see \cite{Ni97} for an account up to 1997 and \cite{Ph15,AF17,AF24a,AF24b} for recent developments). The study of algebraic relations between values of $M$-functions at some non-zero algebraic points -- possibly associated with distinct $p$ or evaluated at distinct algebraic points -- eventually reduces to the study, for each $p$, of the algebraic relations between some $M_p$-functions.
One tool to perform this task is the Galois theory of Mahler equations.
However, despite recent developments, the  Galois theory does not allow to treat general equations of order greater than $3$ (see \cite{DP23}). 
In the shift case, Feng \cite{Feng18} provides an algorithm that applies to equations of arbitrary order, but it is inefficient\footnote{It seems that Feng's algorithm could potentially be adapted to the case of Mahler equations; see the discussion in \cite[p.\,1226]{AF24b}.}. 
As the Galois group on an equation acts on the solutions, the understanding of the nature of these solutions in a neighborhood of $0$ can be of some help in computing it. In general, the form of the solutions of Mahler equations may be rather complicated and involve Hahn series with intricate supports. It becomes simpler when the equation is \textit{regular singular} at $0$, for it only involves Puiseux series and solutions of equations with constant coefficients. 

In light of the above, it becomes important to provide simple criterion for an equation to be regular singular at $0$. This is what is achieved in this paper, in a general setting.
In \cite{FP} the authors provided a first algorithm to determine whether a Mahler equation is regular singular at $0$ or not\footnote{Actually, this algorithm is written for Mahler systems, but it can be used for Mahler equations by considering  the associated companion matrix.}. However, the algorithm was not really informative about the properties of regular singular Mahler equations. In particular, it did not provide any simple criterion for an equation to be regular singular at $0$. 
In this paper, we provide a new algorithm\footnote{We implemented the algorithm in Python 3, it is available at the following URL address:\\ \url{https://faverjon.perso.math.cnrs.fr/AlgoRS_Newton.py}.} which improves the results of \cite{FP} in three directions: it runs faster, it is much more informative about the nature of regular singular Mahler equations and it does not restrict to equations with polynomial coefficients. 

\subsection*{Main results}
Let $\K$ be an algebraically closed field of characteristic $0$.
Let $\Puis$ be the field of Puiseux series, with coefficients in $\K$, that is,
$$
\Puis := \bigcup_{d =1}^\infty \K\left(\left(z^{1/d}\right)\right).
$$ 
Let $p\geq 2$ be an integer. We consider the operator
$$\begin{array}{rccl}
\puip :  & \Puis & \rightarrow & \Puis \\
         & f(z) & \mapsto & f\left(z^p\right)\, ,
\end{array}$$
which acts trivially on $\K$.
A \emph{$p$-Mahler equation} of order $m\in\mathbb{N}^\star$ over $\Puis$, or, for short, a \emph{Mahler equation} is a linear functional equation of the form
\begin{equation}\label{eq:Mahler_at_0}
a_m(z)f(z^{p^m})+ \cdots + a_1(z)f(z^p) + a_0(z)f(z) =0\,,
\end{equation}
with $a_0(z),\ldots,a_m(z) \in \Puis$ and $a_0(z)a_m(z) \neq 0$. To each $p$-Mahler equation we can associate an operator
\begin{equation}\label{eq:form_operateur} 
L := a_m\puip^m + \ldots + a_1\puip + a_0 \in \Puis\langle \puip \rangle\,.
\end{equation}
Then, \eqref{eq:Mahler_at_0} may be rewritten as $Lf=0$. 
A \textit{$p$-Mahler system} over $\Puis$ is a system of linear equations of the form
$$
\puip(Y)=AY
$$
with $A \in {\rm GL}_m(\Puis)$. To any $p$-Mahler equation we can associate a $p$-Mahler system by considering the companion matrix
\begin{equation}\label{eq:companion_matrix}
A_L:=\begin{pmatrix}
 0   &1 &&0
    \\\vdots&\ddots&\ddots 
    \\0&\cdots&0&1
    \\ -\frac{a_0}{a_m} & \cdots & \cdots &-\frac{a_{m-1}}{a_m}
\end{pmatrix}\,.
\end{equation}
Conversely, to any system we can associate an equation, using the Cyclic Vector Lemma (see \cite[Section 3.1]{FP} for more details).

\begin{defi} Let $L \in \Puis\langle \puip \rangle$. The $p$-Mahler equation $Lf=0$ is \textit{regular singular at $0$} if there exists a matrix $P \in {\rm GL}_m(\Puis)$ such that $\puip(P)^{-1}A_LP \in {\rm GL}_m(\K)$. 
\end{defi}

\begin{thm}\label{thm:algo_complexity}
   Let $\K$ be a computable algebraically closed field of characteristic $0$.
  Algorithm \ref{Algo:main}, described in Section \ref{sec:algo_main}, returns whether a $p$-Mahler equation of the form \eqref{eq:Mahler_at_0} is regular singular at $0$. If the $a_i$'s are power series, it runs in $$\mathcal O\left(m^2\nu^2p^m \right)$$ 
 operations in $\K$, 
  where $\nu$ is the maximum of the valuations of the $a_i$'s.
\end{thm}
Recall that a field is said to be \textit{computable} if its elements can be effectively represented and there exist algorithms to perform basic arithmetic operations and decide equality (see \cite{Rabin60}). Note that the algebraic closure of a computable field is itself computable \cite[Theorem 7]{Rabin60}. For example, one could take $\K = \Q$ in Theorem \ref{thm:algo_complexity}.
Note that the Algorithm \ref{Algo:main} will only use a finite number of coefficients from the $a_i$'s, depending only on the ramification of the $a_i$'s and their valuations. Thus, we do not need a way to compute arbitrary large coefficients of the $a_i$'s but only to know their coefficients up to an effective bound.

Our algorithm relies on the Frobenius method for Mahler equations introduced in \cite{Ro23}. It uses the Newton polygon of the Mahler equation \eqref{eq:Mahler_at_0}, which is the lower convex hull of the set of points $(p^{i},j)$, $0\leq i \leq m $ and $j \geq \val a_{i}(z)$, see Sections \ref{sec:Newton} and \ref{sec:Frobenius}. When considering linear ordinary differential equations, an equation is regular singular at $0$ if and only if its Newton polygon has only one slope and this slope is $0$ \cite[p.92]{VDPS03}. There is no such criterion for Mahler equations: the regular singularity at $0$ cannot be read from the Newton polygon. Precisely, as shown in Remarks \ref{rem:inverse} and \ref{rem:heuristic}, one may find two distinct equations having the same Newton polygon such that the first equation is regular singular at $0$ while the second is not. Nevertheless, some information can be extracted from the Newton polygon.

\begin{thm}\label{thm:slopes}
    A necessary condition for a Mahler equation with $a_i \in \K[[z]]$ to be regular singular at $0$ is that the denominators of the slopes of its Newton polygon are relatively prime with $p$. 
\end{thm}

\begin{rem}\label{rem:PowerSeries}
    We can always reduce our study of Equation \eqref{eq:Mahler_at_0} to the situation where the $a_i$'s are power series. Indeed, replacing $z$ with $z^d$ for some $d$ and multiplying the equation with some monomial does not affect the regular singularity at $0$.  
\end{rem}
The criterion given by Theorem \ref{thm:slopes} becomes particularly nice when $p$ is large.

\begin{coro}
\label{coro:sensdirect}
Consider an equation \eqref{eq:Mahler_at_0} with $a_i \in \K[[z]]$. Assume that the equation is regular singular at $0$ and that $p > \max_{i} \val a_i$. Then, one of the following holds:
\begin{itemize}
    \item the Newton polygon has one unique slope; 
    \item the Newton polygon has two slopes and the second one is null.
\end{itemize}
\end{coro}

Using this criterion and the algorithm cited in Theorem \ref{thm:algo_complexity} we are able to answer one of the questions we asked in \cite{FP}: is that true that a $p$-Mahler equation is either regular singular at $0$ for every $p$ large enough or not regular singular at $0$ for every $p$ large enough? We give a positive answer as follows: 

\begin{thm}\label{thm:p_varies}
   Let $a_0,\ldots,a_m \in \K[[z]]$ with $a_0a_m\neq 0$ and set $\nu := \max_i \val a_i$. For any $p> \nu$ consider the $p$-Mahler equation 
   \begin{equation}\label{eq:mahler_2}\tag{\theequation-$p$}
  a_m(z)f(z^{p^m})+ \cdots + a_1(z)f(z^p) + a_0(z)f(z) =0\,.
   \end{equation}
   The following two propositions are equivalent:
   \begin{itemize}
       \item there exists a $p> \nu$ such that the $p$-Mahler equation \eqref{eq:mahler_2} is regular singular at $0$;
       \item the $p$-Mahler equation \eqref{eq:mahler_2} is regular singular at $0$ for any $p> \nu$.
   \end{itemize}
\end{thm}

\subsection*{Organisation of the paper}
Our paper is organized as follows. In Section \ref{sec:extensions}, we introduce some ring extensions of the field of Puiseux series on which the map $\puip : z \mapsto z^p$ extends. In Section \ref{sec:Newton} we introduce the Newton polygon associated with a Mahler equation. We briefly recall the Frobenius method for Mahler equations of \cite{Ro23} (which is an adaptation of the corresponding method for linear differential equations) in Section \ref{sec:Frobenius}. The proof of Theorem \ref{thm:slopes} is given in Section \ref{sec:slopes}. In Section~\ref{sec:RSiff}, we give a necessary and sufficient condition for an equation to be regular singular at $0$, from which the algorithm of Theorem \ref{thm:algo_complexity} is built. The main algorithm, namely Algorithm \ref{Algo:main}, is described in Section \ref{sec:algo_main} where we also prove Theorem \ref{thm:algo_complexity}. In Section \ref{sec:exemple}, we run our main algorithm on an example. In Section \ref{sec:oneslope} we apply our results to the case where the Newton polygon associated with the equation has one slope. Last, in Section \ref{sec:p_grand} we prove Corollary \ref{coro:sensdirect} and Theorem \ref{thm:p_varies}.

\subsubsection*{Acknowledgment} We are deeply grateful to the referee for their meticulous reading of an earlier version of this paper.

\section{Some useful ring extensions of the field of Puiseux series}\label{sec:extensions}

The equation $y(z^p)-2y(z)=0$ shows that a Mahler equation may not have a basis of solutions in the ring $\Puis$ of Puiseux series. To build a basis of solutions, one has to consider some ring extension of $\Puis$. There exists a formal construction for such an extension following from the Picard-Vessiot theory. In the case of Mahler equations, one may give a more precise description of a basis of solutions.

\subsection{Hahn series} \label{sec:Hahn}
Let $R$ be a commutative ring. A \textit{Hahn series} $f$ is a formal series 
$$
f(z)=\sum_{\gamma \in \mathbb Q} f_\gamma z^\gamma,\quad f_\gamma \in R
$$
whose support, $\supp f:= \{\gamma \in \mathbb Q \, : \, f_\gamma \neq 0\}$, is a well-ordered subset of $\mathbb Q$. We let $\Hahn_R$  denote the ring of Hahn series with coefficients in $R$. It contains the ring $\Puis_R$ of Puiseux series with coefficients in $R$. When $R=\K$, we let $\Hahn=\Hahn_\K$. Equipped with the termwise addition and the Cauchy product, $\Hahn_R$ is a ring.
When $R$ is a field, $\Hahn_R$ is a field too (see \cite[Section 3.1]{Ro23}). We define the \textit{valuation} of a Hahn series $f(z)=\sum_{\gamma} f_\gamma z^\gamma$ as $\val f= \min (\supp f)$ with the convention that $\val 0 = +\infty$. It is a valuation which extends the valuation on $\Puis_R$. We define the \textit{coefficient of lowest degree} of $f$ as $\cld f= f_{\val f}$. The map $\puip$ naturally extends to $\Hahn$ and we have, for any $f \in \Hahn$, $\val \puip(f) = p\val f$.  

Some Hahn series, like the series $\sum_{k\geq 0} z^{-1/p^k}$, are solutions of $p$-Mahler equations\footnote{A description of Hahn series which are solutions of Mahler equations is given in \cite{FR24}.}. Yet, the field of Hahn series is not large enough to contain a basis of solutions of any $p$-Mahler equation over $\Puis$.

\subsection{The ring $\Rspe$} 
Recall that $\K$ is algebraically closed. Consider a family of indeterminates $X_c$, $c \in \K^\times$, along with an additional indeterminate $Y$. The ring $\Hahn[(X_c)_{c\in \K^\times},Y]$ can be equipped with a structure of $\puip$-extension of $\Hahn$, by defining the action of $\puip$ as follows: 
	$$
	\forall c \in \K^\times, \,  \puip(X_c)=cX_c \quad \text{  and  }\quad \puip(Y)=Y+1.$$
One easily checks that the ideal $I$ spanned by $X_c-1$ and the polynomials $X_cX_d-X_{cd}$, for all $c,d\in \K^\times$, is invariant under $\puip$. Thus, the quotient ring
	$$\Rspe=\Hahn[(X_c)_{c\in \K^\times},Y]/I$$
inherits a structure of $\puip$-ring. Precisely, if we let $e_c$ denote the image of $X_c$ for each $c\in \K^\times$ and $\logm$ denote the image of $Y$, then $$
\puip(e_c)=ce_c,\, \forall c \in \K^\times, \text{ and }\puip(\logm)=\logm+1.
$$
Note that, with these notations, $e_1=1$.

A straightforward adaptation of the proof of Theorem 35 in \cite{Ro18}\footnote{In \cite{Ro18}, this property is proved over $\Puis$ and the base field is $\Q$, but the same argument works over $\Hahn$ and any algebraically closed field of constants, $\K$.} shows that for any $g \in \Rspe$, the identity $\puip(g)=g$ holds if and only if $g \in \K$. Furthermore, it follows from \cite[Remark 7 and Theorem 12]{Ro23} that any $p$-Mahler equation \eqref{eq:Mahler_at_0} has a basis of solutions in $\Rspe$, meaning that it has $m$ solutions that are linearly independent over $\K$. Note that the base field $\K$ is $\C$ in \cite{Ro23}  but the argument there works for any algebraically closed field of characteristic zero.
	
Any element $y$ of $\Rspe$ can be written
$$
y = \sum_{c,j} h_{c,j}e_c\logm^j
$$
with $h_{c,j} \in \Hahn$ and where the sum is over a finite range of $c \in \K^\star$ and $j\in \Z_{\geq 0}$.
Thus, any $p$-Mahler equation has $m$ solutions $y_1,\ldots,y_m$ of the form
$
y_i= \sum_{c,j} h_{i,c,j}e_c\logm^j
$,
with $h_{i,c,j} \in \Hahn$, which are linearly independent over $\K$.

\begin{prop}\label{prop:RS_equations}
Consider a $p$-Mahler equation \eqref{eq:Mahler_at_0}. The following are equivalent:
\begin{enumerate}[label=\alph*)]
\item it is regular singular at $0$;
\item it has a basis of solutions of the form $
\sum_{c,j} h_{i,c,j}e_c\logm^j$
with $h_{i,c,j}\in \Puis$;
\item any solution in $\Rspe$ of this equation is of the form
$\sum_{c,j} h_{c,j}e_c\logm^j$ with $h_{c,j} \in \Puis$, that is,  any solution belongs to $\Puis[(e_c)_{c \in \K^\star},\logm]$.
\end{enumerate}
\end{prop}
\begin{proof}  Let us prove that b) implies c). Any solution is a linear combination over $\K$ of elements of the basis. Thus, it is of the form
    $\sum_{c,j} h_{c,j}e_c\logm^j$
with $h_{c,j} = \sum_{i=1}^m b_i h_{i,c,j}$ and $b_i \in \K$. Since the $h_{i,c,j}$'s are Puiseux series, so are the $h_{c,j}$'s. 

Let us prove that a) implies b). Assume that the system is regular singular at $0$, so there exists $P \in {\rm GL}_m(\Puis)$ such that $A_0:=\puip(P)^{-1}A_LP \in {\rm GL}_m(\K)$. It follows from \cite[Lemma 24]{FR24} that there exists an invertible matrix $E$, with coefficients in $\K[(e_c)_{c \in \K^\star},\logm]$ such that 
\begin{equation}
\label{eq:pourC}
    \puip(E)^{-1}A_0E={\rm I}_m\,.
\end{equation}
Thus, 
$$
\puip(PE)=A_LPE\,.
$$
It follows that the $i$th column of $PE$ is of the form $(y_i,\puip(y_i),\ldots,\puip^{m-1}(y_i))^\top$ for some $y_i \in \Puis[(e_c)_{c \in \K^\star},\logm]$ and that $y_1,\ldots,y_m$ are solutions of \eqref{eq:Mahler_at_0}. Since the matrice $PE$ is invertible, these solutions are linearly independent. Thus, $y_1,\ldots,y_m$ form a basis of solutions and b) holds.

We prove that c) implies a). Assume that c) holds. In \cite[Theorem 2]{Ro20}, Roques proves that there exists $H \in {\rm GL}_m(\Hahn)$ such that $A_0:=\puip(H)^{-1}A_LH \in~{\rm GL}_m(\K)$. From \cite[Lemma 24]{FR24}, there exists an invertible matrix $E$ with coefficients in $\K[(e_c)_{c \in \K^\star},\logm]$ which satisfies \eqref{eq:pourC}. Without loss of generality, we can assume that $A_0$ is a Jordan matrix so that we can choose $E$ to be upper triangular with diagonal entries $e_{c_1},\ldots,e_{c_m}$, where $c_1,\ldots,c_m$ are the diagonal entries of $A_0$. Since 
$
\puip(HE)=A_LHE
$
 and since $A_L$ is a companion matrix, the entries of the $i$th row of $HE$ are of the form $\puip^{i-1}(f)$ where $f$ is a solution of \eqref{eq:Mahler_at_0}. In the meantime, the entries of $HE$ are of the form 
$$
\sum_{c,j} h_{c,j}e_c\logm^j
$$
with $h_{c,j} \in \Hahn$. It follows from c) that $h_{c,j} \in \Puis$. Furthermore, since the matrix $E$ is upper triangular, the entries of $H$ are in $\Puis$.  Thus the system is regular singular at $0$ and a) holds.
\end{proof}

\subsection{The ring $\Rspe_\lambda$} \label{sec:Rlambda}

Let $\lambda$ be an indeterminate, which will serve as a parameter in applying the Frobenius method to Mahler equations. This method is an analog of the eponymous method for differential equations (see Section \ref{sec:Frobenius}). We consider the field $\Hahn_{\K(\lambda)}$  (respectively $\Puis_{\K(\lambda)}$) of Hahn series (respectively Puiseux series) in $z$ whose coefficients are rational functions in $\lambda$. We let the map $\puip$ act on $\Hahn_{\K(\lambda)}$, leaving $\lambda$ unchanged. Furthermore the derivative $\partial_\lambda = \frac{\partial}{\partial \lambda}$ acts on  $\Hahn_{\K(\lambda)}$. Note that $\puip$ and $\partial_\lambda$ commute.

By \cite[Section 3.4]{Ro23}, there exists a differential ring extension $\Rspe_\lambda$ of $\Hahn_{\K(\lambda)}$ on which $\puip$ and $\partial_\lambda$ act and commute and which has an element $e_\lambda$ such that
$$
\puip(e_\lambda)=\lambda e_\lambda
$$
and that $\Rspe_\lambda =\Hahn_{\K(\lambda)}[e_\lambda,\partial_\lambda e_\lambda,\partial_\lambda^2e_\lambda,\ldots]$. Here, the elements $\partial_\lambda^i(e_\lambda)$, $i \in \Z_{\geq 0}$, are considered as indeterminates and $\partial_\lambda(\partial_\lambda^ie_\lambda) = \partial_\lambda^{i+1}e_\lambda$ for all $i \in \Z_{\geq 0}$.

Let $c \in \K^\star$. We consider the ring $\K[[\lambda-c]]^{\rm rat} \subset \K(\lambda)$ of rational functions with no poles at $\lambda=c$.  There is a specialization map $\ev_{\lambda = c}$ from the subring $\Hahn_{\K[[\lambda-c]]^{\rm rat}}[e_\lambda,\partial_\lambda e_\lambda,\partial_\lambda^2e_\lambda,\ldots]$ of $\Rspe_\lambda$ to $\Rspe$ which sends $\lambda$ to $c$, $e_\lambda$ to $e_c$ with the property that it commutes with $\puip$. Let $\logm_{c,k}=\ev_{\lambda=c}(\partial_\lambda^k e_\lambda/k!)$. Still following \cite{Ro23}, we can define the specializations so that
\begin{equation}\label{eq:def_log_ck}
\logm_{c,k} = \frac{1}{c^k} \frac{\logm(\logm-1)\cdots (\logm-k+1)}{k!}e_c\, ,
\end{equation}
with the convention $\logm_{c,0} = e_c$.

\section{Newton Polygons}\label{sec:Newton}

Let $L$ be the operator attached to \eqref{eq:Mahler_at_0}, defined by \eqref{eq:form_operateur}. Following \cite{CDDM18}, we define the Newton polygon $\mathcal N(L)$  of the equation \eqref{eq:Mahler_at_0} 
as the lower convex hull of the set
$$
\Px(L)= \{(p^{i},j) \ \vert \ i \in \{0,\ldots,m\}, \ j  \geq \val a_{i}(z)\} \subset \R^{2}. 
$$
The polygon $\mathcal N(L)$ is delimited by two vertical half lines and by $\kappa\leq m-1$ non-vertical edges
having pairwise distinct slopes $\mu_1<\mu_2<\cdots < \mu_{\kappa}$, called the slopes  of Equation~\eqref{eq:Mahler_at_0}.  In this section, we gather some properties of Newton polygons which will be used in the proofs of Theorems \ref{thm:algo_complexity} and \ref{thm:slopes}.

\begin{lem}\label{lem:slopes_valuation}Let $f\in\mathcal{H}$ be a non-zero solution of \eqref{eq:Mahler_at_0}. Then $\val f$ is the opposite of a slope of $\mathcal N(L)$.
\end{lem}
\begin{proof}
    See \cite[Lemma 8]{FR}.
\end{proof}
Each edge of $\mathcal N(L)$ is delimited by two points $(p^i,\val a_i)$ and $(p^j,\val a_j)$ with $0\leq i <j\leq m$. The slope of such an edge is equal to
$$
\frac{\val a_j - \val a_i}{p^j-p^i}\,.
$$
The integer $j-i$ is called the \textit{multiplicity} of this slope. Let $r_1,\ldots,r_\kappa$ denote the respective multiplicities of $\mu_1,\ldots,\mu_{\kappa}$. To each slope $\mu_j$, we associate a characteristic polynomial $\chi_j(\lambda)$ in the following way. Let $\I_j$ denote the set of $i \in \{0,\ldots,m\}$ such that the point $(p^i,\val a_i)$ belongs to the edge with slope $\mu_j$. We set
$$
\chi_j(\lambda)=\sum_{i \in \I_j} \cld a_i(z) \lambda^i\, ,
$$
where $\cld a_i(z)$ is the coefficient of lowest degree of $a_i$ (see Section \ref{sec:Hahn}).
We let $c_{j,1},\ldots,c_{j,r_j} \in \K^\star$ denote the nonzero roots of $\chi_j(\lambda)$ counted with multiplicities: they are the \textit{exponents} attached to the slope $\mu_j$. We let $m_{c,j}$ denote the multiplicity of an exponent $c$ attached to the slope $\mu_j$. If $c$ is not attached to $\mu_j$ we write $m_{c,j}=0$. Note that $\sum_{c,j} m_{c,j}=m$. For any pair $(c,j)$ we set
$$
s_{c,j}:=m_{c,1}+\cdots + m_{c,j-1}\,.
$$

\begin{ex}
\label{ex:NewtonPolyg2}
Let $p=2$ and $L= 2z^2\puip^4 +(1+z)\puip^3 + (-2+z^2)\puip^2 + (1-z)\puip +2z^3 $.
The slopes are $\mu_1 = - 3$ with multiplicity $r_1=1$, $\mu_2=0$ with multiplicity $r_2=2$ and $\mu_3 = \frac{1}{4}$ with multiplicity $r_3=1$. The associated characteristic polynomials are respectively
$$
\chi_1(\lambda) = \lambda+2,\quad \chi_2(\lambda)= \lambda^3-2\lambda^2+\lambda\quad \text{ and } \quad \chi_3(\lambda)=2\lambda^4+\lambda^3\,.
$$
Thus, the exponents are $c_{1,1}=-2$, $c_{2,1}=c_{2,2}=1$ and $c_{3,1}=-2$ and we have
$m_{-2,1}=1$, $m_{1,2}=2$, $m_{-2,3}=1$ and  $m_{c,j}=0$  for every other pairs $(c,j)$.

\begin{figure}[h!]
\centering
\definecolor{rfsqsq}{rgb}{0.12,0.12,0.12}
\begin{tikzpicture}[line cap=round,line join=round,>=triangle 45,x=0.6cm,y=0.6cm]
\clip(-1,-0.7) rectangle (18,4);
\fill[line width=1pt,color=rfsqsq,fill=rfsqsq,fill opacity=0.1] (1,3) -- (2,0) -- (8,0) -- (16,2) -- (16,3.8) -- (1,3.8) -- cycle;
\draw [line width=1pt] (1,3)-- (2,0);
\draw [line width=1pt] (2,0)-- (8,0);
\draw [line width=1pt] (8,0)-- (16,2);
\draw [line width=1pt] (16,2)-- (16,3.8);
\draw [line width=1pt] (1,3) -- (1,3.8);
\draw [line width=1pt,color=rfsqsq] (1,3)-- (2,0);
\draw [line width=1pt,color=rfsqsq] (2,0)-- (8,0);
\draw [line width=1pt,color=rfsqsq] (8,0)-- (16,2);
\draw [line width=1pt,color=rfsqsq] (16,2)-- (16,3.8);
\draw [line width=1pt,color=rfsqsq] (1,3.8)-- (1,3);
\draw [->,line width=1pt] (0,0) -- (17.2,0);
\draw [->,line width=1pt] (0,0) -- (0,3.8);
\draw (0.8,-0.02) node[anchor=north west] {$1$};
\draw (-0.8,1.4) node[anchor=north west] {$1$};
\draw (1.8,-0.02) node[anchor=north west] {$2$};
\draw (3.8,-0.02) node[anchor=north west] {$4$};
\draw (7.8,-0.02) node[anchor=north west] {$8$};
\draw (15.8,-0.02) node[anchor=north west] {$16$};
\begin{scriptsize}
\draw [fill=black] (1,3) circle (1.5pt);
\draw [fill=black] (2,0) circle (1.5pt);
\draw [fill=black] (2,1) circle (1.5pt);
\draw [fill=black] (4,2) circle (1.5pt);
\draw [fill=black] (4,0) circle (1.5pt);
\draw [fill=black] (8,0) circle (1.5pt);
\draw [fill=black] (8,1) circle (1.5pt);
\draw [fill=black] (16,2) circle (1.5pt);
\draw (1,0) node {$\vert$};
\draw (16,0) node {$\vert$};
\draw (0,1) node {\textemdash};
\end{scriptsize}
\end{tikzpicture}
\caption{\small \textit{The Newton polygon $\mathcal N(L)$ from Example \ref{ex:NewtonPolyg2}}}
\end{figure}
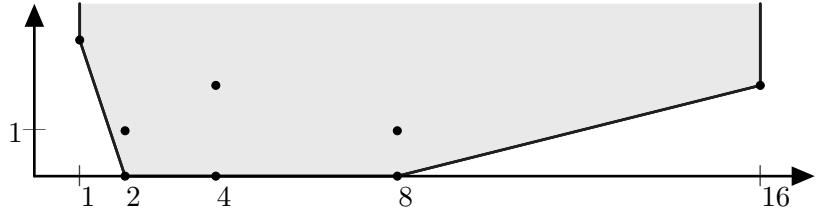
\end{ex}

For any $c \in \K^\star$, we define the operator $L_c$ as follows:
$$
L_c=a_mc^m\puip^m+\cdots + a_1c\puip+a_0 \in \Puis\langle \puip \rangle\,.
$$
\begin{rem}\label{rem_Lc} Let $c \in \K^\star$. The operator $L_c$ has the following properties. 
\begin{itemize}
    \item We have $L_cf=0$ if and only if $L(fe_c)=0$. In particular, it follows from Proposition \ref{prop:RS_equations} that $L$ is regular singular at $0$ if and only if $L_c$ is.
    \item The Newton polygons of  $L$ and $L_c$ are identical. In the meantime, the exponents attached to a slope of $L_c$ are obtained from the ones attached to the same slope of $L$ by dividing by $c$.
\end{itemize}
\end{rem}
Similarly, we define
$$
L_\lambda=a_m\lambda^m\puip^m+\cdots + a_1\lambda\puip+a_0 \in \Puis_{\K(\lambda)}\langle \puip \rangle\,.
$$
Once again, $L_\lambda f=0$ if and only if $L(fe_\lambda)=0$ for any $f \in \Rspe_\lambda$. 
Note the following result for latter use.

\begin{lem}\label{lem:chi}
For any integer $j$ such that $1\leq j \leq \kappa$, we have 
$$
\cld L_\lambda(z^{-\mu_j}) = \chi_j(\lambda)\,.
$$
Furthermore, if $-v\in\mathbb{Q}$ is not a slope of $\mathcal{N}(L)$, then 
$\cld L_\lambda(z^{v}) = \lambda^{i_0}\cld (a_{i_0})$ for a certain $i_0\in\lbrace 0,\ldots, m\rbrace$.

\end{lem}
\begin{proof}
   By definition of $L_\lambda$ we have 
    $L_\lambda(z^{v}) = \sum_{i=0}^m a_i(z)\lambda^i z^{p^iv}$ for any $v \in \mathbb Q$.
    Let $\gamma = \min_{0\leq i\leq m}\{p^iv+\val a_i\}$ and let $\I=\{i\,:\,p^iv+\val a_i=\gamma\}$. Then,
    \begin{equation}\label{eq:image_monome}
    \cld( L_\lambda(z^{v}) )=\sum_{i\in\I}  \cld(a_i(z) \lambda^iz^{p^iv})=\sum_{i\in\I} \cld a_i(z) \lambda^i\,.
    \end{equation}
    If $v=-\mu_j$ for some integer $j\in\{1,\ldots,\kappa\}$, by definition of $\I_j$ we have $\I=\I_j$ and the result follows from \eqref{eq:image_monome} and the definition of $\chi_j(\lambda)$. On the contrary, when $-v$ is not a slope of $\mathcal N(L)$, the set $\I$ is a singleton and the result again follows from \eqref{eq:image_monome}.
  \end{proof}

For an operator $L =\sum_{i=0}^m a_i\puip^i\in \Puis_{\K(\lambda)}\langle \puip \rangle$, $a_0a_m \neq 0$, we define similarly the Newton polygon $\mathcal N(L)$ and its slopes $\mu_1<\cdots < \mu_{\kappa}$.
\begin{lem}
\label{lem:valuation_L(f)}
Consider an operator $L=\sum_{i=0}^m a_i\puip^i\in \Puis_{\K(\lambda)}\langle \puip \rangle$, $a_0 \neq 0$, and let $\mu_1<\cdots<\mu_\kappa$ denote the slopes of $\mathcal N(L)$. Let $f\in\Hahn_{\K(\lambda)}$ be such that $\val(f)>-\mu_1$. Then, 
$$
\val\big( L(f)\big) = \val(a_0)+ \val(f) \quad \text{ and }\quad \cld\big( L(f)\big) = \cld(a_0)\times  \cld(f) \, . 
$$
\end{lem}

\begin{proof}  Let $v:=\val(f)$.
The valuation of each term $a_i\puip^i(f)$ in $L(f)$ is given by $\val(a_i)+p^iv$ for $i=0,\ldots,m$. For any $i>0$, since $\mu_1$ is the smallest slope, we have
$$
\frac{\val (a_i) - \val (a_0)}{p^i-1}\geq \mu_1.
$$
Thus,
$$
\val(a_i)+p^iv-(\val(a_0) + v)\geq (p^i-1)(v+\mu_1)>0.
$$
This proves the first equality. The second one is an immediate consequence.
\end{proof}

\begin{lem}   \label{lem:solutions_inhomo_Puiseux}
    Consider an operator $L=\sum_{i=0}^m a_i\puip^i\in \Puis_{\K(\lambda)}\langle \puip \rangle$ with $a_0 \neq 0$.
    Let $a_\infty\in\Puis_{\K(\lambda)}$ be such that $\val (a_\infty) > \val (a_0) - \mu_1$. Then, the equation $Ly=a_\infty$ has a solution in $f\in\Puis_{\K(\lambda)}$ such that $\val f>-\mu_1$.
\end{lem}

\begin{proof}
Making a change of variables $z \mapsto z^d$ if necessary, we may assume without loss of generality that the $a_i$'s belong to $\K(\lambda)((z))$. Let $\gamma=\val (a_\infty)-\val (a_0)>-\mu_1$. 
 For a Laurent series $g=\sum_{k} g_kz^k$ we let $[g]_k=g_k$.
Define recursively a Laurent series $f=\sum_{k\geq \gamma} f_kz^k$ by
\begin{equation*}\label{eq:sol_inhomo_rec}
f_{\gamma} = \frac{\cld a_\infty}{\cld a_0},\qquad f_{\gamma+n+1}  =\frac{1}{\cld a_0}\left[ a_\infty - L\left(\sum_{k=\gamma}^{\gamma+n}f_{k}z^k\right)\right]_{\val a_\infty+n+1}.
\end{equation*}
One first check by induction that, for any $n \geq 0$,
\begin{equation}\label{eq:val_ainfty}
\val \left(a_\infty - L\left(\sum_{k=\gamma}^{\gamma+n}f_{k}z^k\right) \right)\geq \val a_\infty + n +1\,.
\end{equation}
\textit{Proof for $n=0$.} When $n=0$, \eqref{eq:val_ainfty} is rewritten as $\val \left(a_\infty - L\left(f_{\gamma}z^\gamma\right) \right)\geq\val a_\infty +1$. This inequality follows from the definitions of $\gamma$ and $f_\gamma$, and from Lemma \ref{lem:valuation_L(f)} applied to $f_\gamma z^\gamma$. 

\noindent \textit{Proof of the inductive step.} Let $n\geq 0$ for which \eqref{eq:val_ainfty} holds. By Lemma \ref{lem:valuation_L(f)}, we have $\val(L(f_{\gamma+n+1}z^{\gamma+n+1})= \val a_\infty+n+1$ and
$$ 
\left[ a_\infty - L\left(\sum_{k=\gamma}^{\gamma+n}f_{k}z^k\right)\right]_{\val a_\infty+n+1} -\left[ L\left(f_{\gamma+n+1}z^{\gamma+n+1}\right)\right]_{\val a_\infty+n+1} =0\,.
$$
We infer from this equality and the induction hypothesis that \eqref{eq:val_ainfty} holds for $n+1$. This proves the inductive step. Thus \eqref{eq:val_ainfty} holds for any $n \geq 0$. Letting $n$ tends to $\infty$ we obtain that $\val(a_\infty - L(f))=+\infty$, that is $L(f)=a_\infty$, as wanted.
\end{proof}

\section{Frobenius Method}\label{sec:Frobenius}
We continue with the notations of the previous sections. 
For any $j\in\{1,\ldots,\kappa\}$, set $\theta_j=\val L_\lambda(z^{-\mu_j})$. It follows from \cite{Ro23} that, for any exponent $c \in \K$ attached to the slope $\mu_j$, there exists an unique $g_{c,j}(\lambda,z) \in \Hahn_{\K(\lambda)}$ such that
\begin{equation}\label{eq:g_cj}
L(g_{c,j}(\lambda,z)e_\lambda)=z^{\theta_j}(\lambda - c)^{s_{c,j}+m_{c,j}} e_\lambda\,.
\end{equation}
Furthermore, $g_{c,j}(\lambda,z)$ is well-defined at $\lambda=c$ and it has valuation $-\mu_j$ in $z$. Moreover,  the coefficient of lowest degree $r_{c,j}(\lambda):=\cld g_{c,j}(\lambda,z)$ is a rational function in $\lambda$ whose $(\lambda-c)$-adic valuation is equal to $s_{c,j}$. It can be written explicitly:
 \begin{equation}\label{eq:cld}
    r_{c,j}(\lambda)=\lambda^{-r_1-\cdots - r_{j-1}}\frac{\prod_{i=1}^j\prod_{k=1}^{r_i}(-c_{i,k})}{a_{0,\val a_0}} \frac{(\lambda - c)^{s_{c,j}}}{\prod_{c'\neq c} (\lambda-c')^{m_{c',j}}} \in \K[[\lambda - c]]^{\rm rat}.
        \end{equation}
Note that \eqref{eq:g_cj} may be rewritten as
\begin{equation}\label{eq:g_cj-bis}
L_\lambda(g_{c,j}(\lambda,z))=z^{\theta_j}(\lambda - c)^{s_{c,j}+m_{c,j}}\,.
\end{equation}
\begin{ex} We continue with the operator $L$ from Example \ref{ex:NewtonPolyg2}. Associated with the first slope and the exponent $-2$ we have $r_{-2,1}(\lambda)=1$, $\theta_1=6$ and 
    $$
    g_{-2,1}(\lambda,z)=z^3+\frac{1}{2}\lambda z^4-\frac{1}{4}\lambda^2z^5+\frac{1}{4}\lambda^2z^6+\frac{1}{8}\lambda^3z^7-\frac{1}{8}\lambda^3z^8 + \text{higher degree terms in $z$}.
    $$
    Further examples of series $g_{c,j}$ may be found in \cite[Section 6]{Ro23}.
\end{ex}
We recall the following result of Roques (see \cite[Theorem 12]{Ro23}).
\begin{thm} \label{thm:Roques}
    The functions 
    $$
   y_{c,j,i}= \ev_{\lambda=c}\left( \partial_\lambda^{s_{c,j}+i}(g_{c,j}(\lambda,z)e_\lambda) \right)
    $$
    where, for each $(c,j)$, $0\leq i <m_{c,j}$, form a basis of solutions of \eqref{eq:Mahler_at_0}.
\end{thm}
Actually, since we derive at most $s_{c,j}+m_{c,j}-1$ times $g_{c,j}(\lambda,z)e_\lambda$ and we specialize at $\lambda=c$, one does not need an equality in \eqref{eq:g_cj} to obtain a basis of solutions. 

\begin{thm}\label{th:Roques_variant}
    Assume that for any $j$ and any exponent $c$ attached to the slope $\mu_j$ we have a Hahn series $h_{c,j}(\lambda,z) \in \Hahn_{\K[[\lambda-c]]^{\rm rat}}$ which has valuation $-\mu_j$ in $z$, such that $\cld h_{c,j}$ has a $(\lambda-c)$-adic valuation equal to $s_{c,j}$, and such that
    \begin{equation}\label{eq:h_cj}
    L(h_{c,j}(\lambda)e_\lambda) \in  (\lambda - c)^{s_{c,j}+m_{c,j}}e_\lambda \Hahn_{\K[[\lambda-c]]^{\rm rat}}\,.
    \end{equation}
    Then the functions
        $$
   y_{c,j,i}= \ev_{\lambda=c}\left( \partial_\lambda^{s_{c,j}+i}(h_{c,j}(\lambda,z)e_\lambda) \right)
    $$
    where, for each $(c,j)$, $0\leq i <m_{c,j}$, form a basis of solutions of \eqref{eq:Mahler_at_0}.
\end{thm}

\begin{proof}
Since $\ev_{\lambda=c}$ and $\partial_\lambda$ commute with $\puip$, the fact that $y_{c,j,i}$ are solutions follows from \eqref{eq:h_cj} and the fact that, for any $\theta(\lambda,z) \in \Hahn_{\K[[\lambda-c]]^{\rm rat}}$,
$$
\ev_{\lambda=c}\left(\partial_\lambda^{s_{c,j}+i} \left( (\lambda - c)^{s_{c,j}+m_{c,j}}e_\lambda\theta(\lambda,z) \right)\right) = 0\,,
$$
when $0 \leq i < m_{c,j}$. Let us prove that the functions are linearly independent over $\K$. Recall that we set $\logm_{c,k}=\ev_{\lambda = c}(\partial_\lambda^ke_\lambda)/k!$. 
The product formula for the derivative implies that
$$
y_{c,j,i} = \sum_{k=0}^{s_{c,j}+i}  \frac{(s_{c,j}+i)!}{(s_{c,j}+i-k)!}\ev_{\lambda=c}\left( \partial_\lambda^{s_{c,j}+i-k}(h_{c,j}(\lambda,z))\right) \logm_{c,k}.
$$
Assume that there exist $\tau_{c,j,i} \in \K$  such that $\sum_{c,j,i} \tau_{c,j,i}y_{c,j,i}=0$. From \cite{Ro23}, the $\logm_{c,k}$'s are linearly independent over $\Hahn$. Thus, for all $c,k$,
\begin{equation}\label{eq:lin_dep_h}
\sum_{j} \sum_{i= \max\{0,k-s_{c,j}\}}^{m_{c,j}-1} \tau_{c,j,i}\frac{(s_{c,j}+i)!}{(s_{c,j}+i-k)!}\ev_{\lambda=c}\left( \partial_\lambda^{s_{c,j}+i-k}(h_{c,j}(\lambda,z))\right)=0\,.
\end{equation}
Furthermore, by assumption on $\val h_{c,j}$ and $\cld h_{c,j}$, 
\begin{equation}\label{eq:val_derivee}
\val \ev_{\lambda=c}\left( \partial_\lambda^{n}(h_{c,j}(\lambda,z))\right) \left\{\begin{array}{ll}
\geq  -\mu_j& \text{ if } n > s_{c,j}, \\ = -\mu_j & \text{ if } n=s_{c,j}, \\ > -\mu_j & \text{ if }n<s_{c,j}\, . \end{array}\right. 
\end{equation}
Fix an exponent $c$. We recall that the indices $(c,j,i)$ of the numbers $\tau_{c,j,i}$ are such that $0\leq j\leq \kappa$ and $0\leq i<m_{c,j}$. We extend this notation for $j=\kappa +1$ setting $\tau_{c,\kappa+1,i}=0$ for all $i$.
We prove by decreasing induction on $j\in\lbrace 1, \ldots, \kappa+1\rbrace$ the property 
$$
(\mathcal{P}_{c,j}) : \quad 
\forall k\in \lbrace 0, \ldots, m_{c,j} -1 \rbrace, \hspace{0.15cm}\tau_{c,j,k}=0 \, .
$$
\textit{Proof of the base case $j=\kappa+1$.} The property $(\mathcal{P}_{c,\kappa+1})$ is satisfied by definition.

\noindent\textit{Proof of the inductive step.} Fix a $j_0\leq \kappa$ and suppose that $(\mathcal{P}_{c,j})$ is satisfied when $j>j_0$. Let us prove $(\mathcal{P}_{c,j_0})$ by decreasing induction on $k$. 

\textit{Proof of the base case $k=m_{c,j_0}-1$}. By \eqref{eq:val_derivee}, any term of \eqref{eq:lin_dep_h} but the one corresponding to $i=m_{c,j_0}-1$ has valuation greater than $-\mu_{j_0}$, since the $\tau_{c,j,k}$'s are null when $j>j_0$. It follows that $\tau_{c,j_0,m_{c,j_0}-1}=0$. 

\textit{Proof of the inductive step}. Let $k<m_{c,j_0}-1$. Suppose that $\tau_{c,j_0,l}=0$ for any $l$ with $k<l\leq m_{c,j_0}-1$. Then, \eqref{eq:lin_dep_h} can be rewritten as
\begin{multline*}
\sum_{j<j_0}\sum_{i= \max\{0,k-s_{c,j}\}}^{m_{c,j}-1} \tau_{c,j,i}\frac{(s_{c,j}+i)!}{(s_{c,j}+i-k)!}\ev_{\lambda=c}\left( \partial_\lambda^{s_{c,j}+i-k}(h_{c,j}(\lambda,z))\right)
\\ + \sum_{i= \max\{0,k-s_{c,j_0}\}}^{k} \tau_{c,j_0,i}\frac{(s_{c,j_0}+i)!}{(s_{c,j_0}+i-k)!}\ev_{\lambda=c}\left( \partial_\lambda^{s_{c,j_0}+i-k}(h_{c,j_0}(\lambda,z))\right)
=0\,.
\end{multline*}
Then, all the terms but the one corresponding to $j=j_0$ and $i=k$ have valuation greater than $-\mu_{j_0}$. It follows that $\tau_{c,j_0,k}=0$. 

By induction on $k$, we have proved that $(\mathcal{P}_{c,j_0})$ is satisfied. 

By induction on $j$, $(\mathcal{P}_{c,j})$ holds for any $j$. Since $c$ was chosen arbitrarily, we eventually get that $\tau_{c,j,k}=0$ for any $c,j,k$, which ends the proof.
\end{proof}

\section{A necessary condition: proof of Theorem \ref{thm:slopes}}\label{sec:slopes}

Theorem \ref{thm:slopes} gives a necessary condition for a $p$-Mahler equation with $a_i \in \K[[z]]$ to be regular singular at $0$: the denominators of the slopes of its Newton polygon must be coprime with $p$.  Before proving it, we first recall the following preliminary fact.

\begin{lem}\label{lem:solution_denom}
    Consider a Mahler equation of the form \eqref{eq:Mahler_at_0} for which $a_i \in \K[[z]]$ and which is regular singular at $0$. Let $d$ be the least common multiple of the denominators of the slopes of $\mathcal N(L)$ which are relatively prime with $p$. Then, any solution of \eqref{eq:Mahler_at_0} is of the form
    $$
    y=\sum_{c,j} h_{c,j}e_c\logm^j
    $$
    where the sum has finite range and where $h_{c,j} \in \K((z^{1/d}))$.
\end{lem}
\begin{proof}
The lemma follows from \cite[Corollary 2.4]{FP} applied to the companion system associated with \eqref{eq:Mahler_at_0}. The fact that the integer $d$ can be chosen as the least common multiple of the denominators of the slopes of $\mathcal N(L)$ follows from the proof of \cite[Lemma 2.3]{FP}. Note that, in \cite{FP} the equation is supposed to have coefficients in $\Q(z)$, however, the proof works all the same over $\K((z))$. 
\end{proof}

\begin{proof}
    [Proof of Theorem \ref{thm:slopes}] 
    Suppose that $\mu_j$ is a slope whose denominator is not coprime with $p$, for some $j$. Let $c$ be an exponent attached to $\mu_j$ and let 
   $$
   f(z)= \ev_{\lambda=c}\left( \partial_\lambda^{s_{c,j}}(g_{c,j}(\lambda,z)e_\lambda) \right),
   $$
   where the $g_{c,j}$'s are defined in Section \ref{sec:Frobenius}.
   Using the product formula, we obtain,
     $$ f(z)=\sum_{k=0}^{s_{c,j}}\frac{s_{c,j}!}{(s_{c,j}-k)!} \ev_{\lambda=c}\left(\partial_\lambda^{s_{c,j}-k}(g_{c,j}(\lambda,z))\right)\logm_{c,k}\, ,
   $$
where we let $\logm_{c,k}=\ev_{\lambda=c}(\partial_\lambda^k e_\lambda/k!)$.
   Let $f_k(z)=\frac{s_{c,j}!}{(s_{c,j}-k)!} \ev_{\lambda=c}(\partial_\lambda^{s_{c,j}-k}(g_{c,j}(\lambda,z)))$ so that $f=\sum_{k=0}^{s_{c,j}} f_k\logm_{c,k}$. From \eqref{eq:cld},  $\val (f_{0})=-\mu_j$ and $\val(f_k)>-\mu_j$ when $0<k\leq s_{c,j}$.
   From \eqref{eq:def_log_ck}, $\logm_{c,0} = e_c$ and for all $k\geq 1$, $\logm_{c,k} $ is a linear combination over $\K$ of $e_c\ell,e_c\ell^2,\ldots,e_c\ell^k$. Thus,
   $$ f=f_0e_c+h_1e_c\logm+\cdots+h_{s_{c,j}}e_c\logm^{s_{c,j}}
   $$
   where $h_1,\ldots,h_{s_{c,j}}$ are linear combinations over $\K$ of $f_1,\ldots,f_k$. Since $f_0$ has valuation $-\mu_j$ and the denominator of $-\mu_j$ is not relatively prime with $p$, it follows from Lemma \ref{lem:solution_denom} that the equation is not regular singular at $0$.
   \end{proof}

\section{A necessary and sufficient condition}\label{sec:RSiff}

One could hope that a necessary and sufficient condition for the equation to be regular singular at $0$ would be for the $g_{c,j}$'s to be Puiseux. This is not the case (see Remark \ref{rem:g_cj}). However, we prove in this section that a necessary and sufficient condition for the equation to be regular singular at $0$ is that the first coefficients of each $g_{c,j}\in \Hahn[[\lambda-c]]$ are Puiseux series. Let us fix some notations and definitions. 

\begin{nota*}
Let $c \in \K$, $f\in \Hahn_{\K[[\lambda-c]]}$ and let $s$ be an integer. We let 
$f \mod (\lambda - c)^s$
denote the remainder of the division of $f$ by $(\lambda-c)^s$ in $ \Hahn_{\K[[\lambda-c]]}$, that is, if  
$f=\sum_{\gamma\in\mathbb Q}\sum_{n \geq 0} f_{\gamma,n}z^\gamma (\lambda-c)^n$, then
$$
f \mod (\lambda - c)^s=\sum_{\gamma\in\mathbb Q}\sum_{n =0}^{s-1} f_{\gamma,n}z^\gamma (\lambda-c)^n\, .
$$
It is a polynomial with degree at most $s-1$ in $\lambda$.

Given an element $f \in  \Hahn[\lambda]$, we let ${\rm val}_{\lambda-c}f$ denote the valuation of $f$ as a polynomial in $\lambda-c$, with coefficients in $\Hahn$. We also let $\deg_{\lambda} f$ denote the degree of $f$ as a polynomial in $\lambda$ (which is the same as the degree of $f$ as a polynomial in $\lambda-c$). 
\end{nota*}

\begin{defi}\label{def:restriction}
We consider the following assumptions about Mahler equations of the form \eqref{eq:Mahler_at_0}:
\begin{enumerate}[label=($\mathcal A_\arabic*$)]
    \item\label{A1} $a_0,\ldots,a_m$ belong to $\K[[z]]$;
    \item\label{A2} the denominators of the slopes of the Newton polygon associated with \eqref{eq:Mahler_at_0} are coprime with $p$.
\end{enumerate}
\end{defi}
Our study of Mahler equations can always be reduced so that \ref{A1} holds (see Remark~\ref{rem:PowerSeries}). Furthermore, it follows from Theorem \ref{thm:slopes} that \ref{A2} is a necessary condition for the equation to be regular singular at $0$. Thus, these two assumptions do not limit the scope of the following results.

In this section, we let \eqref{eq:Mahler_at_0} be a $p$-Mahler equation satisfying \ref{A1} and \ref{A2} and $L$ be the associated operator given by \eqref{eq:form_operateur}. Let $d$ denote the least common multiple of the denominators of the slopes $\mu_1<\cdots<\mu_\kappa$ of the associated Newton polygon. Recall that $r_{c,j}(\lambda)$ is defined for any pair $(c,j)$ by \eqref{eq:cld}.
\begin{defi}
    \label{defi:quasi_solution}
     Let $1\leq j \leq \kappa$ be an integer and $c$ be an exponent attached to the slope $\mu_j$. We say that $f_{c,j}(\lambda,z)=\sum_v f_{c,j,v}(\lambda)z^v \in \K[z^{\pm \frac{1}{d}},\lambda]$ is a \textit{truncated solution associated with $(c,j)$} if the following five conditions hold:
     \medskip
     
      \begin{enumerate}[label=($\mathcal C_\arabic*$)]
\item\label{C1} $\val \left(L_\lambda(f_{c,j}) \mod (\lambda - c)^{s_{c,j}+m_{c,j}}\right) > \val(a_0) -\mu_1$,
\medskip

\item\label{C2} $\supp f_{c,j} \subset \frac{1}{d}\Z \cap [-\mu_j, -\mu_1]$,
\medskip
     
\item\label{C3} $\cld f_{c,j} = r_{c,j}(\lambda) \mod (\lambda -c)^{s_{c,j}+m_{c,j}}$,

     \medskip

\item\label{C4}$\val f_{c,j}=-\mu_j$,
\medskip

\item\label{C5}$\deg_{\lambda} f_{c,j}\leq s_{c,j}+m_{c,j}-1$.
\end{enumerate}
We say that a truncated solution is \textit{reduced} if the following additional condition holds:
\begin{enumerate}[label=($\mathcal C_\arabic*$)]
    \setcounter{enumi}{5}
\item\label{C6} $\deg_\lambda f_{c,j,-\mu_k}(\lambda)\leq s_{c,j}+m_{c,j}-m_{c,k}-1$, for every $k<j$.
    \end{enumerate}
\end{defi}
The following theorem is the core of Algorithm \ref{Algo:main}.

\begin{thm}\label{thm_RS_si_solutions_tronquees}
  Let \eqref{eq:Mahler_at_0} be a $p$-Mahler equation satisfying \emph{\ref{A1}} and \emph{\ref{A2}}. The following are equivalent :
  \begin{itemize}
      \item[$({\rm i})$] the equation is regular singular at $0$;
      \item[$({\rm ii})$] for any $j$ and any exponent $c$ attached to $\mu_j$ there exists a truncated solution $f_{c,j}$ associated with $(c,j)$;
      \item[$({\rm iii})$] for any $j$ and any exponent $c$ attached to $\mu_j$ there exists a reduced truncated solution $f_{c,j}$ associated with $(c,j)$.
  \end{itemize}
\end{thm}
 
Theorem \ref{thm_RS_si_solutions_tronquees} is proved after two lemmas.
\begin{lem}\label{lem:cond_v}
Let $f_{c,j}=\sum_v f_{c,j,v}(\lambda)z^v \in \K[z^{\pm \frac{1}{d}},\lambda]$ satisfying Condition \emph{\ref{C5}} in Definition \ref{defi:quasi_solution}. Then, the following are equivalent:
\begin{itemize}
    \item[(a)] \emph{\ref{C6}} holds;
    \item[(b)] $\deg_\lambda f_{c,j,-\mu_k}(\lambda)\leq s_{c,j}+m_{c,j}-m_{c,k}-1$, for every $k<j$ such that $c$ is attached to $-\mu_k$;
    \item[(c)] $\deg_\lambda f_{c,j,v}(\lambda)\leq s_{c,j}+m_{c,j}-{\rm val}_{\lambda-c} (\cld L_\lambda(z^v))-1$ for every $v> -\mu_j$.
\end{itemize}
\end{lem}
\begin{proof}
By definition, $m_{c,k}\neq 0$ if and only if $c$ is attached to $-\mu_k$. Using \ref{C5}, it gives the equivalence between $(a)$ and $(b)$. Furthermore, it follows from Lemma \ref{lem:chi} that, for any $v \in \mathbb Q$,
$$
{\rm val}_{\lambda-c} (\cld L_\lambda(z^v)) = \left\{\begin{array}{ll} 
{\rm val}_{\lambda-c}(\chi_k(\lambda))= m_{c,k} & \text{if } v=-\mu_k \text{ for some }k,\\
{\rm val}_{\lambda-c} (\lambda^{i_0} \cld(a_{i_0}))=0 & \text{for some $i_0$, else}.
\end{array}\right.
$$
Thus, $(a)$ and $(c)$ are equivalent.
\end{proof}
\begin{lem}\label{lem:solution_g_lambda}
 Let $c$ be an exponent attached to some slope $\mu_j$. Then,
 $$
 \ev_{\lambda=c}\left( \partial_\lambda^{i}(g_{c,j}(\lambda,z)e_\lambda) \right)
 $$
 is a solution of \eqref{eq:Mahler_at_0} for any integer $i$ such that $0 \leq i < s_{c,j}+m_{c,j}$.
\end{lem}
\begin{proof}
It is immediate with \eqref{eq:g_cj} and the fact that  $\partial_\lambda$ and ${\rm ev}_{\lambda=c}$ commute with $\puip$.
\end{proof}
\begin{proof}[Proof of Theorem \ref{thm_RS_si_solutions_tronquees}]
Let us first prove that (i) implies (iii). Suppose that the equation is regular singular at $0$.
    For each pair $(c,j)$, $1 \leq j \leq \kappa$, such that $c$ is an exponent attached to the slope $\mu_j$ we let $g_{c,j}$ be given by Equation \eqref{eq:g_cj}. Write
    $$
    g_{c,j}(\lambda,z) = \sum_{n=0}^\infty g_{c,j,n}(z)(\lambda-c)^n\,.
    $$
The proof that (i) implies (iii) is divided in three facts.

\medskip
\noindent    \textbf{Fact 1.} \textit{For any triplet $(c,j,n)$ with $0 \leq n < s_{c,j}+m_{c,j}$, $g_{c,j,n}(z)\in \K((z^{1/d}))$.}

 \medskip
  Fix a pair $(c,j)$. We are going to prove Fact 1 by induction on $n$.

\noindent\textit{Proof of the base case $n=0$.} It follows from Lemma \ref{lem:solution_g_lambda} that $g_{c,j,0}(z)e_c$ is a solution of \eqref{eq:Mahler_at_0}. Thus, since the equation is regular singular at $0$, it follows from Lemma \ref{lem:solution_denom} that $g_{c,j,0}(z)\in\K((z^{1/d}))$.

\noindent\textit{Proof of the inductive step.} Let $n\geq 1$ and suppose that $g_{c,j,i}(z)\in \K((z^{1/d}))$ for any $i<n$. Since $\partial_\lambda$ and $\ev_{\lambda=c}$ commute, we have
\begin{align*}
    \ev_{\lambda=c}\left( \partial_\lambda^{n}(g_{c,j}(\lambda,z)e_\lambda) \right)& = \sum_{i=0}^n \binom{n}{i} \ev_{\lambda=c}\left(\partial_\lambda^{i}(g_{c,j}(\lambda,z))\right)\ev_{\lambda=c}\left(\partial_\lambda^{n-i}e_\lambda\right)
    \\ & = \sum_{i=0}^n n!g_{c,j,i}(z) \logm_{c,n-i}
\end{align*}
Thus, using \eqref{eq:def_log_ck}, there exists some non-zero $\gamma_{i,k}\in  \Q$, $0 \leq i \leq n$, $0 \leq k\leq n-i$, such that
$$
\ev_{\lambda=c}\left( \partial_\lambda^{n}(g_{c,j}(\lambda,z)e_\lambda) \right) = \sum_{k=0}^n\left(\sum_{i=0}^{n-k}\gamma_{i,k}g_{c,j,i}(z)\right) \logm^k e_c\,.
$$
Since $n\leq s_{c,j}+ m_{c,j}-1$, it follows from Lemma \ref{lem:solution_g_lambda} that $    \ev_{\lambda=c}\left( \partial_\lambda^{n}(g_{c,j}(\lambda,z)e_\lambda) \right)$ is a solution of \eqref{eq:Mahler_at_0}. Since the equation is regular singular at $0$, it follows from Lemma \ref{lem:solution_denom} that, for every $k$, the series $\sum_{i=0}^{n-k}\frac{\gamma_{i,k}}{c^{n-i}}g_{c,j,i}$ belongs to $\K((z^{1/d}))$. Taking $k=0$, and using the induction hypothesis, we obtain that $g_{c,j,n} \in \K((z^{1/d}))$, as wanted. This proves Fact 1.

\medskip
For each $n$ with $0\leq n\leq s_{c,j}+ m_{c,j}-1$ write
$$
g_{c,j,n}(z)= \sum_{\gamma \geq -\mu_j} g_{c,j,n,\gamma}z^\gamma\, \text{ and set }\, h_{c,j,n}(z)=\sum_{-\mu_j\leq \gamma \leq-\mu_1} g_{c,j,n,\gamma}z^\gamma\,,
$$
where all the $\gamma$'s involved in the sums above belong to $\frac{1}{d}\Z$. Set also
$$
h_{c,j}(\lambda,z):=\sum_{n=0}^{s_{c,j}+m_{c,j}-1} h_{c,j,n}(z)(\lambda-c)^n\,.
$$
The Puiseux polynomial $h_{c,j}$ can be seen as the reduction of $g_{c,j}$ modulo $(\lambda-c)^{s_{c,j}+m_{c,j}}$ and modulo $z^{-\mu_1+1/d}$.

\medskip
\noindent\textbf{Fact 2.} \textit{For any pair $(c,j)$, $h_{c,j}(\lambda,z)$ is an associated truncated solution.}

\medskip
By construction, $h_{c,j}\in \K[z^{\pm \frac{1}{d}},\lambda]$ and it has degree at most $s_{c,j}+m_{c,j}-1$ in $\lambda$. Hence $h_{c,j}$ satisfies \ref{C5}. The support of $h_{c,j}$, as a Puiseux polynomial in $z$, is included in $\frac{1}{d}\Z\cap [-\mu_j ; -\mu_1]$. Thus, $h_{c,j}$ satisfies \ref{C2}. Let $\theta = g_{c,j}-h_{c,j}$. By construction, we can write 
$$\theta=\theta_1+\theta_2(\lambda-c)^{s_{c,j}+m_{c,j}},$$
with $\val \theta_1 >-\mu_1$. Thus,
\begin{equation*}\label{eq:L_lambda_hcj}
L_\lambda(h_{c,j}) = L_\lambda(g_{c,j})- L_\lambda(\theta)=L_\lambda(g_{c,j})-L_\lambda(\theta_1) - (\lambda-c)^{s_{c,j}+m_{c,j}}L_\lambda(\theta_2)
\end{equation*}
Reducing modulo $(\lambda - c)^{s_{c,j}+m_{c,j}}$ and using \eqref{eq:g_cj-bis}
$$
L_\lambda(h_{c,j}) \mod  (\lambda - c)^{s_{c,j}+m_{c,j}} = -L_\lambda(\theta_1) \mod  (\lambda - c)^{s_{c,j}+m_{c,j}} \,.
$$
Then, 
\begin{align*}
\val \left(L_\lambda(h_{c,j}) \mod  (\lambda - c)^{s_{c,j}+m_{c,j}}\right) &= \val\left(L_\lambda(\theta_1) \mod  (\lambda - c)^{s_{c,j}+m_{c,j}} \right)
\\ &>\val(a_0)-\mu_1
\end{align*}
because $\val\left(L_\lambda(\theta_1)\right)  = \val(a_0)+\val(\theta_1)$ by Lemma \ref{lem:valuation_L(f)}.
Thus, $h_{c,j}$ satisfies \ref{C1}. It remains to prove \ref{C3} and \ref{C4}. 
By construction, $h_{c,j}$ is the truncation of $g_{c,j}$ where we only keep the monomials of the form $(\lambda-c)^i z^j $ with $0\leq i\leq s_{c,j}+m_{c,j}-1$ and $-\mu_j\leq j\leq -\mu_1$. Recall that  $\val g_{c,j} = -\mu_j$ and $\cld g_{c,j}=r_{c,j}(\lambda)$. Since, by \eqref{eq:cld}, $${\rm val}_{\lambda-c}  (r_{c,j}(\lambda)) = s_{c,j}\leq s_{c,j}+m_{c,j}-1,$$ we have $\val h_{c,j} = \val g_{c,j}=-\mu_j$ and $\cld h_{c,j} = r_{c,j}(\lambda) \mod (\lambda - c)^{s_{c,j}+m_{c,j}}.$ This proves \ref{C3} and \ref{C4} and ends the proof of Fact 2.
\medskip

The Puiseux polynomials $h_{c,j}$ do not necessarily satisfy \ref{C6}. Hence they may not be reduced truncated solutions. Our last step to prove that (i) implies (iii) is to build some reduced truncated solutions $f_{c,j}$, using the Puiseux polynomials $h_{c,j}$. 
We now fix a exponent $c$ attached to some slope.

\medskip

\noindent\textbf{Fact 3.}\textit{ For any integer $j$ such that $c$ is attached to the slope $\mu_j$, there exists a reduced truncated solution $f_{c,j}$ associated with $(c,j)$. }

\medskip
We proceed by induction on the integers $j$ such that $c$ is attached to the slope $\mu_j$. Let $j_0$ be the least integer such that $c$ is attached to $\mu_{j_0}$.

\noindent\textit{Proof of the base case $j=j_0$.} In that case, the condition \ref{C6} is empty. Hence the Puiseux polynomial $f_{c,j_0} = h_{c,j_0}$ satisfies \ref{C6}. Thus, it is a reduced truncated solution associated with $(c,j_0)$. 

\noindent\textit{Proof of the inductive step.} Let $j>j_0$ such that $c$ is attached to $\mu_j$. Let $\mathcal K$ be the set of all $k<j$ such that $c$ is attached to the slope $\mu_k$. Assume that there exists a reduced truncated solution $f_{c,k}$ associated with $(c,k)$ for any integer $k\in \mathcal K$. 
Then, for any polynomials $p_k(\lambda)\in \K[\lambda]$, $k \in \mathcal K$, the Puiseux polynomial
\begin{equation}\label{eq:reduction_fcj}
f_{c,j} := h_{c,j}+\sum_{k\in \mathcal K} (\lambda-c)^{s_{c,j}+m_{c,j}-s_{c,k}-m_{c,k}}p_k(\lambda)f_{c,k} \mod (\lambda-c)^{s_{c,j}+m_{c,j}}
\end{equation}
still satisfies Conditions \ref{C1} to \ref{C5}. Indeed, for all $k \in \mathcal K$, 
$$\val f_{c,k} = -\mu_k > - \mu_j = \val h_{c,j}.$$ 
Thus $\val f_{c,j}= \val h_{c,j}$ and $\cld f_{c,j} = \cld h_{c,j}$, which gives Conditions \ref{C3} and \ref{C4}. Condition \ref{C2} immediately follows from the fact that $f_{c,k}$, $k\in \mathcal{K}$, and $h_{c,j}$ satisfies \ref{C2}. Since $f_{c,j}$ is a remainder modulo $(\lambda-c)^{s_{c,j}+m_{c,j}}$,  Condition \ref{C5} is trivially satisfied. Let us establish \ref{C1}. For all $k\in \mathcal{K}$, since $f_{c,k}$ satisfies \ref{C1} and since $L_\lambda$ is $\K[\lambda]$-linear, one has
\begin{multline*}
  \val(L_\lambda((\lambda-c)^{s_{c,j}+m_{c,j}-s_{c,k}-m_{c,k}}p_k(\lambda)f_{c,k}) \mod (\lambda-c)^{s_{c,j}+m_{c,j}})
  \\ \geq  \val(L_\lambda((\lambda-c)^{s_{c,j}+m_{c,j}-s_{c,k}-m_{c,k}}f_{c,k})\mod (\lambda-c)^{s_{c,j}+m_{c,j}})
  \\ =  \val(L_\lambda(f_{c,k})\mod (\lambda-c)^{s_{c,k}+m_{c,k}}) >\val(a_0)-\mu_1 
\end{multline*}
Since $h_{c,j}$ satisfies \ref{C1}, we also have $\val(L_\lambda(h_{c,j}) \mod (\lambda-c)^{s_{c,j}+m_{c,j}}))>\val(a_0)-\mu_1$ and we deduce from \eqref{eq:reduction_fcj} that $f_{c,j}$ satisfies \ref{C1}.

It only remains to prove that we can choose the polynomials $p_k(\lambda)$ so that $f_{c,j}$ also satisfies \ref{C6}. It follows from Lemma \ref{lem:cond_v} that one only has to check Condition \ref{C6} for the integers $k<j$ which belong to $\mathcal K$.  
Let $k_1>\cdots >k_t$ be an enumeration of the elements of $\mathcal K$ and set, for any $n\in\{1,\ldots,t\}$
$$
f_{c,j}^{[n]} := h_{c,j}+\sum_{i=1}^n (\lambda-c)^{s_{c,j}+m_{c,j}-s_{c,k_i}-m_{c,k_i}}p_{k_i}(\lambda)f_{c,k_i} \mod (\lambda-c)^{s_{c,j}+m_{c,j}}
$$
and $f_{c,j}^{[n]}=\sum_{v \in \mathbb Q} f_{c,j,v}^{[n]}(\lambda)z^v$. Note that $\deg_\lambda f_{c,j,-\mu_{k_n}}^{[n]} =\deg_\lambda f_{c,j,-\mu_{k_n}}$ for any $n \in\{1,\ldots,t\}$. Thus, $f_{c,j}$ satisfies \ref{C6} if and only if, for any $n \in \{1,\ldots,t\}$,
\begin{equation}\label{eq:recurrence_reduced}
    \deg_\lambda f_{c,j,-\mu_{k_n}}^{[n]} \leq s_{c,j}+m_{c,j} - m_{c,k_n}-1
\end{equation}
Let $n \in \{1,\ldots,t\}$. Suppose that $p_{k_1},\ldots,p_{k_{n-1}}$ have already been chosen. Let us prove that we can choose $p_{k_n}$ so that \eqref{eq:recurrence_reduced} holds. Since $f_{c,k_n}$ satisfies \ref{C3} and \ref{C4} it follows from \eqref{eq:cld} that the coefficient of $z^{-\mu_{k_n}}$ in $f_{c,k_n}$ is of the form $\theta(\lambda)(\lambda-c)^{s_{c,k_n}}$, with $\theta(\lambda)$ a unit of $\K[[\lambda-c]]$. Write $f_{c,j,-\mu_{k_n}}^{[n-1]} = q_0(\lambda)-q_1(\lambda)(\lambda-c)^{s_{c,j}+m_{c,j}-m_{c,k_n}}$, with $\deg_\lambda(q_0)\leq s_{c,j}+m_{c,j}-m_{c,k_n}-1$ and let $p_{k_n}(\lambda)$ be the reduction modulo $(\lambda-c)^{m_{c,k_n}}$ of $q_1(\lambda)\theta(\lambda)^{-1}$. 
Since 
$$f_{c,j}^{[n]}=f_{c,j}^{[n-1]}+ (\lambda-c)^{s_{c,j}+m_{c,j}-s_{c,k_n}-m_{c,k_n}}p_{k_n}(\lambda)f_{c,k_n}\mod (\lambda-c)^{s_{c,j}+m_{c,j}}$$
with such a choice for $p_{k_n}$, $f_{c,j,-\mu_{k_n}}^{[n]}$ satisfies \eqref{eq:recurrence_reduced}. Thus, we may chose recursively the polynomials $p_{k_1},\ldots,p_{k_t}$ so that \eqref{eq:recurrence_reduced} holds for any $n \in \{1,\ldots,t\}$.
For such a choice, $f_{c,j}$ satisfies \ref{C6}. \textit{A fortiori}, it is a reduced truncated solution associated with $(c,j)$. This proves Fact 3 and the fact that (i) implies (iii).

\medskip
The fact that (iii) implies (ii) is immediate so it only remains to prove that (ii) implies (i). Suppose that, for any pair $(c,j)$ there exists a truncated solution $f_{c,j}$.
We decompose
$$
L_\lambda = \sum_{n=0}^{m} L_{c,n} (\lambda - c)^n,
$$
where the operators $L_{c,n}$ do not depend on $\lambda$. 
We also decompose
$$
f_{c,j}(\lambda,z)=\sum_{n=0}^{s_{c,j}+m_{c,j}-1}f_{c,j,n}(z)(\lambda - c)^n\,.
$$

\noindent\textbf{Fact 4.} \textit{For all non-negative integer $n \leq s_{c,j}+m_{c,j}-1$, there exist $\theta_{c,j,n}(z) \in \Puis$ such that $\val (\theta_{c,j,n}-f_{c,j,n}(z))>-\mu_1$ and 
\begin{equation}\label{eq:L_lambda,n}
\sum_{i=0}^n L_{c,n-i}(\theta_{c,j,i}(z))=0\, .
\end{equation}
 }

\medskip
We proceed by induction on $n$.

\noindent\textit{Proof of the base case $n=0$.} 
Let $n=0$ and let $a_\infty(z)=-L_{c,0}(f_{c,j,0}(z))$. It follows from \ref{C1} that $\val (a_\infty) > \val a_0 - \mu_1$. Then, it follows from Lemma~\ref{lem:solutions_inhomo_Puiseux} that there exists $y_{c,j,0} \in \Puis$, with valuation greater than $-\mu_1$ and such that $L_{c,0}(y_{c,j,0})=a_\infty$. Thus, $\theta_{c,j,0}=f_{c,j,0}+y_{c,j,0}$ has the properties we want. 

\noindent\textit{Proof of the inductive step.} Let $n\geq 1$ and suppose that the assertion is proved with $n-1$. Let 
$$a_\infty(z) = -L_{c,0}(f_{c,j,n}(z))-\sum_{i=0}^{n-1} L_{c,n-i}(\theta_{c,j,i}(z)).$$
Then
$$
a_\infty(z)=-\sum_{i=0}^nL_{c,n-i}(f_{c,j,i}(z)) - \sum_{i=0}^{n-1} L_{c,n-i}(\theta_{c,j,i}(z)-f_{c,j,i}(z)).
$$
Looking at the term in $(\lambda-c)^n$ in Condition \ref{C1}, it follows from \ref{C1} that the first sum has a valuation greater than $\val a_0 - \mu_1$. It is also the case of each term in the second sum because $\val L_{c,n-i}(\theta_{c,j,i}(z)-f_{c,j,i}(z)) \geq \val L_{\lambda}(\theta_{c,j,i}(z)-f_{c,j,i}(z))$ which is greater than  $\val a_0 - \mu_1$ by induction hypothesis and Lemma \ref{lem:valuation_L(f)}.  Thus, $\val a_\infty > \val a_0 - \mu_1$. Then, it follows from Lemma \ref{lem:solutions_inhomo_Puiseux} that there exists $y_{c,j,n} \in \Puis$, with valuation greater than $-\mu_1$ and such that $L_{c,0}(y_{c,j,n})=a_\infty$. Thus, $\theta_{c,j,n}=f_{c,j,n}+y_{c,j,n}$ has the properties we want. 
This proves the inductive step and, by induction, this proves Fact 4.

\medskip
Let $\theta_{c,j,n}$ be given by Fact 4. Since $\val(\theta_{c,j}-f_{c,j}) > - \mu_1$ and $\val f_{c,j}\leq -\mu_1$ (see \ref{C4}), $\val \theta_{c,j}= \val f_{c,j} = -\mu_j$ and $\cld \theta_{c,j} = \cld f_{c,j}$. Then, the $(\lambda-c)$-adic valuation of $\cld \theta_{c,j}$ is $s_{c,j}$ by \ref{C3} and \eqref{eq:cld}. Setting $\theta_{c,j}(\lambda,z)=\sum_{n=0}^{s_{c,j}+m_{c,j}-1} \theta_{c,j,n}(z)(\lambda-c)^n$ and $\theta_{c,j,n}=0$ when $n\geq s_{c,j}+m_{c,j}$, we infer from \eqref{eq:L_lambda,n} that 
\begin{align*}
L(\theta_{c,j}e_\lambda) =e_\lambda \sum_{n=0}^{\infty}\sum_{i=0}^nL_{c,n-i}(\theta_{c,j,i}(z))(\lambda - c)^n \in (\lambda-c)^{s_{c,j}+m_{c,j}}e_\lambda \Hahn_{\K[[\lambda-c]]^{\rm rat}}\,.
\end{align*}
Thus, it  follows from Theorem \ref{th:Roques_variant} that the functions
$$
y_{c,j,i}(z)=\ev_{\lambda=c}\left(\partial_\lambda^{s_{c,j}+i}(\theta_{c,j}(\lambda,z)e_\lambda)\right)\, ,
$$
$1\leq j \leq \kappa$, $c \in \K^\times$ attached to $\mu_j$, $0 \leq i \leq m_{c,j}-1$, form a basis of solutions of \eqref{eq:Mahler_at_0}. Since the $\theta_{c,j}$'s belong to $\Puis[\lambda]$, the equation is regular singular at $0$ (see Proposition \ref{prop:RS_equations}). Thus (ii) implies (i), which ends the proof of Theorem \ref{thm_RS_si_solutions_tronquees}.
\end{proof}

\section{Algorithmic considerations: proof of Theorem \ref{thm:algo_complexity}}\label{sec:algo_main}

The algorithm mentioned in Theorem \ref{thm:algo_complexity} is described at the end of this section, after a first algorithm to compute reduced truncated solutions. Consider the map $\pi$ which, to any $w \in \mathbb Q$ associates the unique number $\pi(w)$ such that $\val L_\lambda(z^{\pi(w)})=w$, that is
\begin{equation}\label{eq:pi}\pi(w) = \max\left\{\frac{w-\val a_i}{p^i}\, : \, 0\leq i \leq m\right\}\footnote{More properties of this map can be found in \cite[Section 3.4]{FR}).}.
\end{equation}

\begin{algorithm}[H]
\SetAlgoLined
\KwIn{An operator $L$ of the form \eqref{eq:form_operateur} such that the equation $Ly=0$ satisfies \ref{A1} and \ref{A2}. A slope $\mu_j$ of $\mathcal N(L)$. An exponent $c$ attached to this slope. The integers $s_{c,j}$ and $m_{c,j}$. The least common multiple $d$ of the denominators of the slopes of $\mathcal{N}(L)$, as defined in Section \ref{sec:RSiff}.}
\KwOut{Whether there exists a reduced truncated solution associated with $(c,j)$.}
Let $r$ be the reduction of \eqref{eq:cld} modulo $(\lambda-c)^{s_{c,j}+m_{c,j}}$.
\\ Set $f:=rz^{-\mu_j}$.
\\ Set $g:=L_\lambda(f) \mod (\lambda-c)^{s_{c,j}+m_{c,j}}$.
\\ Set $v:=\pi(\val g)$.
\\ \While{$v\leq -\mu_1$}{
 \If{$v \notin \frac{1}{d}\Z$}{\Return False.}
\Else{Set $\alpha := \cld L_\lambda(z^v)$ and $\beta := \cld g$.
\\ \If{there exists $h \in \K[\lambda]$ such that
\begin{itemize}
    \item $\deg_\lambda h\leq s_{c,j}+m_{c,j} - 1 - {\rm val}_{\lambda-c}(\alpha)$
    \item $\beta=\alpha h \mod(\lambda-c)^{s_{c,j}+m_{c,j}}$
\end{itemize}}{
Set $f$ to $f-hz^v$.
 \\  Set $g:=L_\lambda(f) \mod (\lambda-c)^{s_{c,j}+m_{c,j}}$.
 \\ Set $v:=\pi(\val g)$.}
 \Else{\Return False.}}}
\Return True.
 \caption{An algorithm to determine the existence of reduced truncated solutions.}\label{Algo:fcj}
\end{algorithm}

\begin{rem*}
In Algorithm~\ref{Algo:fcj}, $f$ represents the construction monomials by monomials of a reduced truncated solution associated with $(c,j)$. 
If there exists a reduced truncated solution $f_{c,j}$ and we are at some intermediate step, we have $\val (f_{c,j}-f) = \pi(\val g)$ where $g=L_\lambda(f) \mod (\lambda-c)^{s_{c,j}+m_{c,j}}$. That is, the monomial we must add to $f$ at the next step is of the form $\gamma z^{\pi(\val g)}$, with $\gamma \in \K[\lambda]$.
\end{rem*}

\begin{prop}\label{prop:algo_fcj} 
  Let $(c,j)$ be such that $c$ is an exponent attached to the slope $\mu_j$. Algorithm \ref{Algo:fcj} determines whether there exists a reduced truncated solution associated with $(c,j)$. Its complexity is
  $$
  \mathcal O(p^mm\nu^2)
  $$
  where $\nu = \max_i\val a_i(z)$.
\end{prop}
\begin{proof}
   Until the end of the proof we let  
   $$(v_i)_{i\geq 1},\,(f_i)_{i\geq 1},\,(g_i)_{i\geq 1},\,(h_i)_{i\geq 1},\,(\alpha_i)_{i\geq 1},\,(\beta_i)_{i\geq 1},$$
   denote the successive values taken by $v,f,g,h,\alpha$ and $\beta$. We first prove that the algorithm stops after finitely many steps. Suppose that the ``While'' loop was already called $i$ times. If either $v_i>-\mu_1$ or $v_i \notin \frac{1}{d}\Z$, the algorithm stops. Assume that $v_i\leq -\mu_1$ and that $v_i \in \frac{1}{d}\Z$. If there exists no $h$ such that $\deg_\lambda h\leq s_{c,j}+m_{c,j} - 1 - {\rm val}_{\lambda-c}(\alpha_i)$ and $\alpha_ih=\beta_i \mod(\lambda-c)^{s_{c,j}+m_{c,j}}$ then the algorithm stops. Suppose on the contrary that such a $h$ exists. Then $h_i=h$ and $f_{i+1}=f_i-h_{i}z^{v_i}$. By definition of the map $\pi$, we have $\val g_i=\val L_\lambda(z^{v_i})$, see \eqref{eq:pi}. Then, there exists $\theta$ such that $\val \theta>\val g_i$ and such that the following holds, modulo $(\lambda-c)^{s_{c,j}+m_{c,j}}$:
    \begin{align*}
    g_{i+1} = L_\lambda(f_{i+1})=g_i -h_{i}L_\lambda(z^{v_i}) = g_i -h_{i}\alpha_{i}z^{\val g_i} + \theta 
=g_i - \cld g_i z^{\val g_i} + \theta\,.
    \end{align*}
    Thus, 
    \begin{equation}
    \label{eq:valgi}
        \val g_{i+1} > \val g_i\quad \text{ and } \quad v_{i+1} = \pi(\val g_{i+1}) > \pi(\val g_i) =v_i\,.
    \end{equation}
    In particular, the sequence $(v_i)_{i\geq 1}$ is increasing. Since the $v_i$'s can only take $d(\mu_j-\mu_1)+1$ distinct values -- the elements of $\frac{1}{d}\Z \cap [-\mu_j,-\mu_1]$ -- the algorithm stops after finitely many steps. Let $t$ be the number of steps. We have $t \leq d(\mu_j-\mu_1)+1$. 

    Let us now prove that the algorithm is correct. Suppose that there exists a reduced truncated solution $f_{c,j}$ associated with $(c,j)$. Let $\gamma_0=-\mu_j<\gamma_1<\cdots<\gamma_s\leq -\mu_1$ be such that $\supp f_{c,j}=\{\gamma_0,\ldots,\gamma_s\}$ and write
   $$
   f_{c,j}=\sum_{k=0}^s f_{c,j,k}(\lambda)z^{\gamma_k}
   $$
   where the polynomials $f_{c,j,k}(\lambda)\in \K[\lambda]$ have degree at most $s_{c,j}+m_{c,j}-1$. We prove by induction on $i \in \{1,\ldots,s+1\}$ that 
   \begin{equation}\label{eq:rec_f_i}
   f_i= \sum_{k=0}^{i-1} f_{c,j,k}(\lambda)z^{\gamma_k}.
   \end{equation}
   It will follow that the number $t$ of steps is equal to $s+1$ and that $\gamma_i=v_i$ for any $i$, $1\leq i \leq s+1$.

  \noindent \textit{Proof of the base case $i=1$.} When $i=1$, \eqref{eq:rec_f_i} follows from \ref{C3} and  \ref{C4}.

   \noindent \textit{Proof of the inductive step.} Suppose that \eqref{eq:rec_f_i} holds for some $i$, $1\leq i \leq s$. Then,
   $$
   f_{i+1}=f_i-h_iz^{v_i}= \sum_{k=0}^{i-1} f_{c,j,k}(\lambda)z^{\gamma_k} - h_iz^{v_i} \mod (\lambda-c)^{s_{c,j}+m_{c,j}}\,.
   $$
   Then, we only have to prove that $v_i = \gamma_i$ and 
   \begin{equation}\label{eq:cld_g}
    h_i=-f_{c,j,i} \mod (\lambda-c)^{s_{c,j}+m_{c,j}}\,.
   \end{equation}
   By induction hypothesis and definition of $g_i$ we have
   \begin{align}
\label{eq:g_i}
\nonumber g_i&= L_\lambda\left(\sum_{k=0}^{i-1} f_{c,j,k}(\lambda)z^{\gamma_k}\right) \mod (\lambda-c)^{s_{c,j}+m_{c,j}}
\\ & = L_\lambda\left(f_{c,j}\right) - L_\lambda\left(\sum_{k=i}^{s} f_{c,j,k}(\lambda)z^{\gamma_k}\right)\mod (\lambda-c)^{s_{c,j}+m_{c,j}}
   \end{align}
From Condition \ref{C6} and Lemma \ref{lem:cond_v}, for any $k>0$, 
$$\deg_{\lambda-c} L_\lambda\left( f_{c,j,k}(\lambda)z^{\gamma_k}\right)=\deg_{\lambda-c} f_{c,j,k}(\lambda)L_\lambda\left(z^{\gamma_k}\right)<s_{c,j}+m_{c,j}.$$
Thus
  \begin{equation}\label{eq:val_Llambda}
\val\left( L_\lambda\left( f_{c,j,k}(\lambda)z^{\gamma_k}\right)\mod (\lambda-c)^{s_{c,j}+m_{c,j}}\right)=\val L_\lambda(z^{\gamma_k})\,.
  \end{equation}
  Thus, by definition of $\pi$,
  $$
  \pi\left(\val\left( L_\lambda\left( f_{c,j,k}(\lambda)z^{\gamma_k}\right)\mod (\lambda-c)^{s_{c,j}+m_{c,j}}\right)\right)=\gamma_k\,.
  $$
  From Condition  \ref{C1} and the definition of $\pi$, for any $k$ with $0 \leq k \leq i-1$, we have
  \begin{equation}\label{eq:pi_val_L_lambda}
  \pi\left(\val \left(L_\lambda(f_{c,j}) \mod  (\lambda-c)^{s_{c,j}+m_{c,j}}\right)\right)>\pi(\val a_0-\mu_1) = -\mu_1 > \gamma_k.
  \end{equation} 
  It follows from \eqref{eq:g_i}, \eqref{eq:val_Llambda} and \eqref{eq:pi_val_L_lambda} that
   $$
  v_i=\pi( \val g_i )=\pi\left( \val \left(L_\lambda\left(\sum_{k=i}^{s} f_{c,j,k}(\lambda)z^{\gamma_k}\right) \mod (\lambda-c)^{s_{c,j}+m_{c,j}}\right)\right)=\gamma_i\,,
   $$  
  and
   \begin{align*}
 \beta_i= \cld   g_i&= -\cld \left( L_\lambda\left(\sum_{k=i}^{s}f_{c,j,k}(\lambda)z^{\gamma_k}\right) \mod (\lambda-c)^{s_{c,j}+m_{c,j}}\right)
 \\&= -\cld \left( L_\lambda(f_{c,j,i}(\lambda)z^{\gamma_i}) \mod (\lambda-c)^{s_{c,j}+m_{c,j}}\right)
 \\&=-f_{c,j,i}(\lambda)\alpha_i \mod (\lambda-c)^{s_{c,j}+m_{c,j}}\,.
   \end{align*}
     Since $\beta_i=h_{i}\alpha_i \mod (\lambda-c)^{s_{c,j}+m_{c,j}}$  we obtain that
  $$
  (f_{c,j,i}+h_{i})\alpha_i \mod (\lambda-c)^{s_{c,j}+m_{c,j}} = 0\,.
  $$
  Since $f_{c,j,i}$ and $h_i$ have degree at most $s_{c,j}+m_{c,j}-{\rm val}_{\lambda-c}(\alpha_i)-1$ in $\lambda-c$, \eqref{eq:cld_g} holds for $i+1$. This proves the inductive step and, by induction, \eqref{eq:rec_f_i} holds for any $i\in \{1,\ldots,s+1\}$.
  
    In the other direction, suppose that the algorithm returns ``True''. Then, it builds a Puiseux polynomial $f=f_t$. At the time the algorithm stops, we have $v =\val g > -\mu_1$, since it does not return ``False''.
    Thus
    $$
    \val\left( L_\lambda(f) \mod (\lambda-c)^{s_{c,j}+m_{c,j}} \right)> -\mu_1
    $$
    and $f$ satisfies Condition  \ref{C1}. By construction it also satisfies Conditions \ref{C2} to \ref{C5}. Eventually it follows from the construction of the polynomials $h$ at each step and from Lemma \ref{lem:cond_v} that it satisfies Condition  \ref{C6}.
   
Now, we compute the complexity.   
   The remainder $r$ can be computed in $\mathcal O(s_{c,j}+m_{c,j})$ which is a $\mathcal O(m)$, for one only has to consider the first $s_{c,j}+m_{c,j}$ terms in the power series expansion of \eqref{eq:cld} in $\lambda-c$. At each step, we only need to compute the coefficients of $L_\lambda(f)$ up to $\val a_0-\mu_1$. For a rational number $v \geq -\mu_j$, the computation of the coefficients of $L_\lambda(z^v)$ up to $\val a_0-\mu_1$ has the same complexity as the number of points of the Newton polygon of $L_\lambda$ which are above the line with slope $-v$ passing through the point $(0,\val a_0-\mu_1)$. 
   This number of points is 
   $$
   \mathcal O(m(\val a_0 - \mu_1) + p^m \mu_j).
   $$
   However, $\val a_0 \leq \nu$, $\mu_1\geq -\frac{\nu}{p-1}$ and $\mu_j\leq \frac{\nu}{p^{m-1}(p-1)}$. Thus, the number of points is in $\mathcal O(m\nu)$.
   Thus, if we keep in memory the computation of the previous $L_\lambda(f)$, the computation of $L_\lambda(f)$ at the next step requires $\mathcal O(m\nu)$ operations. Then, the computation of $v$, $\alpha$, $\beta$ and $\frac{\beta}{\alpha}$ are in $\mathcal O(1)$. Furthermore, there are at most $d(\mu_j-\mu_1)=\mathcal O(d\nu)$ steps.
    Since $d \leq p^m$, this concludes the proof.
\end{proof}

We may now prove Theorem \ref{thm:algo_complexity}. Precisely, we prove that Algorithm \ref{Algo:main}, which is based on Theorem \ref{thm_RS_si_solutions_tronquees}, determines whether a Mahler equation is regular singular or not. We then prove that it has the expected complexity.

\begin{algorithm}[H]
\SetAlgoLined
\KwIn{A $p$-Mahler equation of the form \eqref{eq:Mahler_at_0}.}
\KwOut{Whether this equation is regular singular at $0$ or not.}
\If{some of the $a_i$'s do not belong to $\K[[z]]$}
{Let $\delta$ be such that $a_i \in \K((z^{1/\delta}))$.
\\ Let $v = \min_i\val a_i(z^{\delta})$. 
\\ Run the next steps of the algorithm with the following $p$-Mahler equation
\begin{equation}\label{eq:change_variable}
z^{-v}a_0(z^\delta)f(z)+z^{-v}a_1(z^\delta)f(z^p)+\cdots + z^{-v}a_m(z^\delta)f(z^{p^m})=0\,.
\end{equation}
}
 Set $L$ the operator of the form \eqref{eq:form_operateur} associated with this equation.
 \\ Compute the slopes of $\mathcal N(L)$, the exponents attached to these slopes, their multiplicities and the least common denominator $d$ of their denominators.
\\ \If{the denominator of one of the slopes is not relatively prime with $p$}{\Return False.}
 \For{each $j$ and each exponent $c$ attached to $\mu_j$}{ \If{Algorithm \ref{Algo:fcj} with inputs $(L,\mu_j,c,s_{c,j},m_{c,j},d)$ returns ``False''}{\Return False}
}
 \Return True
 \caption{An algorithm to determine if a Mahler equation is regular singular at $0$.}\label{Algo:main}
\end{algorithm}

\begin{proof}[Proof of Theorem \ref{thm:algo_complexity}]
At the end of the first ``if'' of Algorithm \ref{Algo:main}, the equation we are working with satisfies the condition \ref{A1} given by Definition \ref{def:restriction}. Since the $p$-Mahler equation \eqref{eq:Mahler_at_0} is regular singular at $0$ if and only if the $p$-Mahler equation \eqref{eq:change_variable} is, we might suppose that the initial equation satisfies \ref{A1}.

    Suppose that the algorithm returns ``True''. In particular, the condition in the second ``if'' is not satisfied, that is, the $p$-Mahler equation satisfies \ref{A2} and each call of Algorithm \ref{Algo:fcj} returns ``True''. Then, it follows from Proposition \ref{prop:algo_fcj} that, associated with each pair $(c,j)$, we have a reduced truncated solution. Thus, it follows from Theorem \ref{thm_RS_si_solutions_tronquees} that the equation is regular singular at $0$. Conversely, suppose that the equation is regular singular at $0$. It follows from Theorem \ref{thm:slopes} that the condition in the second ``if'' is not satisfied. In particular, the equation satisfies \ref{A2}. It then follows from Theorem \ref{thm_RS_si_solutions_tronquees} that we can associate a reduced truncated solution with each pair $(c,j)$. Thus, from Proposition \ref{prop:algo_fcj}, each call of Algorithm \ref{Algo:fcj} returns ``True'' and, eventually, Algorithm \ref{Algo:main} returns ``True''.

Suppose that the initial equation satisfies \ref{A1}. The computation of the slopes of the Newton polygon and the associated characteristic polynomials is in $\mathcal O(m)$. Then, since $\K$ is algebraically closed, determining the nonzero roots of these characteristic polynomials and their multiplicities can be done in $\mathcal O(m)$. Checking if the denominator of each slope is relatively prime with $p$ can be done in $\mathcal O(m)$. Suppose that the condtion in the second ``if'' is not satisfied. Then, Algorithm \ref{Algo:main} calls Algorithm \ref{Algo:fcj} at most $m$ times. Thus, it follows from Proposition \ref{prop:algo_fcj} that it has complexity 
    $$
    \mathcal O\left(m^2\nu^2p^m \right).
    $$
\end{proof}

\begin{rem*}
Proposition \ref{prop:RS_equations} can in fact be refined: when the system is regular singular at $0$, there is a basis of solutions of the form $e_c\sum_{j} h_{i,c,j}\ell^j$, with $c$ attached to some slope of the Newton polygon. Then, using the reduced truncated solutions built at each call of Algorithm \ref{Algo:fcj}, one can compute the first coefficients of the Puiseux polynomials $h_{i,c,j}$.
One could also adapt the algorithm of \cite{CDDM18} to perform this task (see the discussion in \cite[p.\,2920--2921]{FP}).
\end{rem*}

\section{Running Algorithm \ref{Algo:main} on an example}\label{sec:exemple}
We propose to run Algorithm \ref{Algo:main} on the $2$-Mahler equation
\begin{equation}\label{eq:ex_inverse_bis}
    z^8f(z^4)-(z^2+z^3+z^7)f(z^2)+(1+z)f(z)=0\,.
\end{equation}
Here, we take $\K=\Q$ and $p=2$. Since the equation has coefficients in $\K[[z]]$, Algorithm \ref{Algo:main} directly computes the slopes of the Newton polygon of \eqref{eq:ex_inverse_bis}, the exponents attached to these slopes and their multiplicities. The Newton polygon is drawn in the left-hand side of the figure below. It has two slopes which are $\mu_1=2$ and $\mu_2=3$. The associated characteristic polynomials are
$$
\chi_1(\lambda)=1-\lambda, \quad \chi_2(\lambda)=-\lambda+\lambda^2\,.
$$
Thus, $c=1$ is the only exponent attached to some slope, and we have $m_{1,1}=m_{1,2}=1$, $s_{1,1}=0$, $s_{1,2}=1$.
Since the least common multiple of the denominators of the slopes is $d=1$ and is relatively prime with $p=2$, Algorithm \ref{Algo:main} calls Algorithm~\ref{Algo:fcj} with $(c,j)=(1,1)$ and $(c,j)=(1,2)$.

\medskip
\noindent \textit{Call of Algorithm \ref{Algo:fcj} with $(c,j)=(1,1)$.}
Algorithm \ref{Algo:fcj} computes the reduction $r$ of $r_{1,1}(\lambda)$ modulo $(\lambda-1)$. 
We have $r=-1$. Then, it sets $f:=rz^{-2}=-z^{-2}$ and
$$
g:=L_\lambda(f) \mod (\lambda - 1) = z^3-1\,.
$$
Then it sets $v:=\pi(\val g)=\pi(0)=0$. Since $v = 0 > - 2=-\mu_1$, Algorithm \ref{Algo:fcj} stops and returns ``True''.

\medskip
\noindent \textit{Call of Algorithm \ref{Algo:fcj} with $(c,j)=(1,2)$.}
Algorithm \ref{Algo:fcj} computes $r_{1,2}(\lambda)$ modulo $(\lambda-1)^2$:
$$
r_{1,2}(\lambda)=\lambda^{-1}\frac{1}{a_{0,0}}\frac{(\lambda - 1)}{1} = \lambda^{-1}(\lambda-1)\equiv\lambda-1 \mod (\lambda-1)^2 \,.
$$
Then, it sets $f:=(\lambda-1)z^{-3}$ and
\begin{align*}
g&:=L_\lambda(f) \mod (\lambda - 1)^2
\\&= (\lambda-1)(\lambda^2z^8z^{-12}-\lambda(z^2+z^3+z^7)z^{-6}+(1+z)z^{-3}) \mod (\lambda - 1)^2
\\ &= (\lambda - 1)(z^{-2}-z) \,.
\end{align*}
Then, it sets $v:=\pi(\val g)=\pi(-2)=-2$. We have $v \leq -\mu_1=-2$ and $v \in \frac{1}{d}\Z$, where $d=1$. Then, the algorithm sets $\alpha:=\cld L_\lambda(z^{-2})=1-\lambda$, $b:=\cld(g)=\lambda-1$. Thus, it sets $h:=-1$ so that $h\alpha=\beta$. Then it replaces $f$ with 
$$f-hz^{v}=(\lambda-1)z^{-3}+z^{-2}\,.$$
Then, is sets $g$ to a new value:
\begin{align*}
g&:=L_\lambda(f) \mod (\lambda - 1)^2 = L_\lambda((\lambda-1)z^{-3}+z^{-2})\mod (\lambda - 1)^2 
\\&= (\lambda - 1)(-z^{-1}+2-z-z^3) +1-z^3,
\end{align*}
and set $v:=\pi(\val g)=\pi(-1)=-1$. We have $v = -1 > -\mu_1=-2$. Thus, Algorithm~\ref{Algo:fcj} stops and returns ``True''.

\medskip
\noindent \textit{End of the execution of Algorithm \ref{Algo:main}.} Since no call to Algorithm \ref{Algo:fcj} has returned ``False'', Algorithm \ref{Algo:main} returns ``True'' and the system is regular singular at $0$.

\begin{rem}\label{rem:inverse}
    In \cite{FP}, we asked for a characterization of operators of the form \eqref{eq:form_operateur} regular singular at $0$ whose \textit{inverse} $a_0\puip^m+a_1\puip^{m-1} +\cdots + a_{m-1}\puip+a_m=0$ also is.  The equation \eqref{eq:ex_inverse_bis}
    is regular singular at $0$, as we have just seen. Using Algorithm \ref{Algo:main} we leave to the reader to check that so is its inverse
     $$
    z^8f(z)-(z^2+z^3+z^7)f(z^2)+(1+z)f(z^4)=0\,.
    $$
 In the meantime, when $\alpha \in \C \setminus\{1\}$, the $2$-Mahler equation
    $$
    z^8f(z^4)-(z^2+z^3+\alpha z^7)f(z^2)+(1+z)f(z)=0
    $$
    has the same Newton polygon than Equation \eqref{eq:ex_inverse_bis} and is regular singular at $0$. However, its inverse
     $$
    z^8f(z)-(z^2+z^3+\alpha z^7)f(z^2)+(1+z)f(z^4)=0\,
    $$
    is not regular singular at $0$, as one could check using Algorithm \ref{Algo:main} (see also the remark below). 
This emphasizes the fact that such a property cannot be read from the Newton polygon.
           \begin{figure}[h!]
\begin{subfigure}[b]{0.48\textwidth}
    \centering
        \resizebox{\linewidth}{!}{
            \definecolor{rfsqsq}{rgb}{0,0,0}
\begin{tikzpicture}[line cap=round,line join=round,>=triangle 45,x=0.7cm,y=0.4cm]
\clip(-3,-1) rectangle (8,9);
\fill[line width=1pt,color=rfsqsq,fill=rfsqsq,fill opacity=0.1] (1,0) -- (2,2) -- (4,8) -- (4,8.8) -- (1,8.8) -- cycle;
\draw [line width=1pt] (1,0)-- (2,2);
\draw [line width=1pt] (2,2)-- (4,8);
\draw [line width=1pt] (1,0) -- (1,8.8);
\draw [line width=1pt] (4,8) -- (4,8.8);
\draw [->,line width=1pt] (0,0) -- (5.5,0);
\draw [->,line width=1pt] (0,0) -- (0,8.8);
\draw (0.8,-0.02) node[anchor=north west] {$1$};
\draw (-0.7,1.7) node[anchor=north west] {$1$};
\draw (1.8,-0.02) node[anchor=north west] {$2$};
\draw (3.8,-0.02) node[anchor=north west] {$4$};
\begin{scriptsize}
\draw [fill=black] (1,0) circle (1.5pt);
\draw [fill=black] (1,1) circle (1.5pt);
\draw [fill=black] (2,2) circle (1.5pt);
\draw [fill=black] (2,3) circle (1.5pt);
\draw [fill=black] (2,7) circle (1.5pt);
\draw [fill=black] (4,8) circle (1.5pt);
\draw (4,0) node {$\vert$};
\draw (2,0) node {$\vert$};
\draw (0,1) node {\textemdash};
\draw (0,2) node {\textemdash};
\draw (0,3) node {\textemdash};
\draw (0,4) node {\textemdash};
\draw (0,5) node {\textemdash};
\draw (0,6) node {\textemdash};
\draw (0,7) node {\textemdash};
\draw (0,8) node {\textemdash};
\end{scriptsize}
\end{tikzpicture}
}
    \end{subfigure}
\begin{subfigure}[b]{0.48\textwidth}
    \centering
        \resizebox{\linewidth}{!}{
            \definecolor{rfsqsq}{rgb}{0,0,0}
\begin{tikzpicture}[line cap=round,line join=round,>=triangle 45,x=0.7cm,y=0.4cm]
\clip(-3,-1) rectangle (8,9);
\fill[line width=1pt,color=rfsqsq,fill=rfsqsq,fill opacity=0.1] (1,8) -- (2,2) -- (4,0) -- (4,8.8) -- (1,8.8) -- cycle;
\draw [line width=1pt] (1,8)-- (2,2);
\draw [line width=1pt] (2,2)-- (4,0);
\draw [line width=1pt] (1,8) -- (1,8.8);
\draw [line width=1pt] (4,0) -- (4,8.8);
\draw [->,line width=1pt] (0,0) -- (5.5,0);
\draw [->,line width=1pt] (0,0) -- (0,8.8);
\draw (0.8,-0.02) node[anchor=north west] {$1$};
\draw (-0.7,1.7) node[anchor=north west] {$1$};
\draw (1.8,-0.02) node[anchor=north west] {$2$};
\draw (3.8,-0.02) node[anchor=north west] {$4$};
\begin{scriptsize}
\draw [fill=black] (1,8) circle (1.5pt);
\draw [fill=black] (4,1) circle (1.5pt);
\draw [fill=black] (2,2) circle (1.5pt);
\draw [fill=black] (2,3) circle (1.5pt);
\draw [fill=black] (2,7) circle (1.5pt);
\draw [fill=black] (4,0) circle (1.5pt);
\draw (1,0) node {$\vert$};
\draw (2,0) node {$\vert$};
\draw (0,1) node {\textemdash};
\draw (0,2) node {\textemdash};
\draw (0,3) node {\textemdash};
\draw (0,4) node {\textemdash};
\draw (0,5) node {\textemdash};
\draw (0,6) node {\textemdash};
\draw (0,7) node {\textemdash};
\draw (0,8) node {\textemdash};
\end{scriptsize}
\end{tikzpicture}
}
    \end{subfigure}

  \caption{\small\textit{The Newton polygons associated with Equation \eqref{eq:ex_inverse_bis} (on the left) and its inverse equation (on the right).}}
\end{figure}
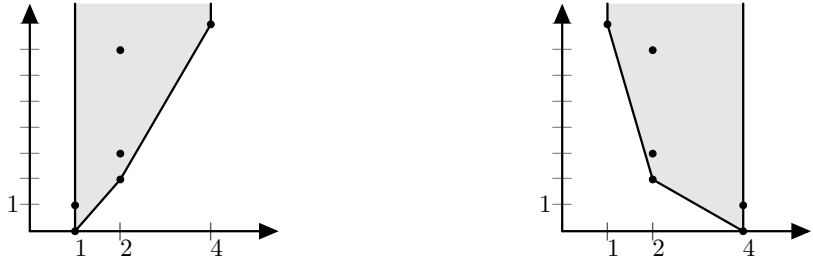
\end{rem}
\begin{rem}\label{rem:heuristic}
   Let us end with some heuristic about regular singular equation, by explaining why the operator $M=
     z^8-(z^2+z^3+\alpha z^7)\phi_2+(1+z)\phi_2^2$ from Remark \ref{rem:inverse} is not regular singular at $0$ when $\alpha \neq 1$ while it is when $\alpha=1$. The associated Newton polygon has two slopes, $-6$ and $-1$ with multiplicity one and exponent $c=1$. We know that there is a truncated solution associated with the first slope, from Lemma~\ref{lem:first_slope} below. Assume that there is a truncated solution associated with the second slope. We choose its coefficients by induction. From Condition \ref{C4}, it would have valuation $1$. For any non-zero $a \in \K[[\lambda-1]]$, 
     $$
     M_\lambda(az^1)\equiv (1-\alpha)az^9 \mod (\lambda- 1).
     $$
     When $\alpha \neq 1$ the first slope of the Newton polygon of the inhomogeneous equation $M(y)=-(1-\alpha)az^9$ is equal to $\frac{7}{2}$. Since its denominator is not coprime with $2$, we cannot build this truncated solution. Thus, the operator is not regular singular at $0$. On the contrary, when $\alpha=1$, such an obstruction does not exists and we can build this truncated solution. Thus, the operator is regular singular at $0$. Since the Newton polygon does not depend on the value of $\alpha$, this emphasizes the fact that, contrary to the differential case, the property of being regular singular cannot be read from the Newton polygon associated with a Mahler equation.
    \end{rem}
    \begin{rem}\label{rem:g_cj}
    Continue with the operator $M$ and set $\alpha=1$ so that the equation is regular singular at $0$. Consider the function $g_{1,2}(\lambda,z)$ associated with $c=1$ and with the second slope $\mu_2=-1$. One can check that
    $$
g_{1,2}(\lambda,z) = \frac{\lambda-1}{\lambda}z+\frac{(\lambda-1)^2}{\lambda}z^{\frac{3}{2}} + \mathcal O(z^2)\,.
    $$
    Since the least common denominator $d$ of the slopes of $\mathcal N(M)$ is equal to $1$ and since $g_{1,2}\notin \K((z))(\lambda)$, $g_{1,2}$ is not a Puiseux series in $z$, although the equations is regular singular at $0$.
    \end{rem}

\section{Newton polygons with one slope}\label{sec:oneslope}
As an application of Theorem \ref{thm_RS_si_solutions_tronquees}, we study here the case of equations whose Newton polygon has only one slope.
We prove that they are always regular singular at $0$.
\begin{prop}\label{prop:one_slope}
   Consider a $p$-Mahler equation over $\Puis$ of the form \eqref{eq:Mahler_at_0} for which the associated Newton polygon has only one slope. Then the Mahler equation is regular singular at $0$.
 \end{prop}

The proposition immediately follows from Theorem \ref{thm_RS_si_solutions_tronquees} and the following lemma.

\begin{lem}\label{lem:first_slope}
   Consider a $p$-Mahler equation over $\Puis$ of the form \eqref{eq:Mahler_at_0}. Let $c$ be an exponent attached to the first slope $\mu_1$ of the associated Newton polygon. Then, the unique solution $g_{c,1}(\lambda,z)$ of Equation \eqref{eq:g_cj} belongs to $\Puis_{\K(\lambda)}$.
\end{lem}
In particular, for any $c$ attached to the first slope, there exists a truncated series $f_{c,1}$ associated with $(c,1)$.
\begin{proof}
By definition of $\theta_1$, introduced at the beginning of Section \ref{sec:Frobenius}, we have $\theta_1=\val a_0 - \mu_1$. Let $r_{c,1}(\lambda) $  be defined by \eqref{eq:cld}. Set 
    $$a_\infty(\lambda,z) = -L_{\lambda}( r_{c,1}(\lambda,z)z^{-\mu_1})+z^{\theta_1}(\lambda - c)^{s_{c,j}+m_{c,j}}.$$
    By \eqref{eq:g_cj-bis}, we have $\val a_\infty(\lambda,z) > \theta_1=\val(a_0(z)) - \mu_1$ and it follows from Lemma~\ref{lem:solutions_inhomo_Puiseux} that the equation
    $  L_{\lambda}(y)=a_\infty(\lambda,z)$ has a Puiseux solution $h \in \Puis_{\K(\lambda)}$.  
Then
     \begin{align*}
     L_{\lambda}(r_{c,1}z^{-\mu_1} + h)&=L_{\lambda}( r_{c,1}(\lambda,z)z^{-\mu_1})+L_{\lambda}(h)
     \\&=L_{\lambda}(r_{c,1}(\lambda,z)z^{-\mu_1})+a_\infty(\lambda,z)
     \\&=z^{\theta_1}(\lambda - c)^{s_{c,j}+m_{c,j}} \,.
     \end{align*}
     Thus
    $r_{c,1}z^{-\mu_1} + h$ satisfies \eqref{eq:g_cj}. Since, from \cite[Proposition 22]{Ro23}, there is only one solution to \eqref{eq:g_cj}, we have $r_{c,1}z^{-\mu_1} + h=g_{c,1}$. In particular, $g_{c,1}(\lambda,z)\in \Puis_{\K(\lambda)}$.
\end{proof}
    
\begin{rem}\label{rem:Fuchsian}
An equation is said to be \textit{Fuchsian at $0$} if the associated companion matrix $A_L$, defined by \eqref{eq:companion_matrix}, is well-defined and non-singular at $0$. It is not difficult to prove that the following assertions are equivalent:
 \begin{itemize}
        \item the equation is Fuchsian at $0$;
        \item the associated Newton polygon has only one slope which is null;
        \item $\val a_0 = \val a_m =\min_{0\leq i \leq m}\val a_i$.
\end{itemize}
\end{rem}

\section{Regular singularity at \texorpdfstring{$0$}{0} for large \texorpdfstring{$p$}{p} }\label{sec:p_grand}

Consider a $p$-Mahler equation \eqref{eq:Mahler_at_0} with $a_i \in \K[[z]]$ and $p> \max_i(\val(a_i))$. Suppose that the equation is regular singular at $0$. According to Corollary \ref{coro:sensdirect}, one of the following holds: 
\begin{itemize} 
\item the Newton polygon has only one slope;
\item it has exactly two slopes, the second one being null.
\end{itemize} 
This section is organized as follows. We first prove Corollary \ref{coro:sensdirect}. The case when the Newton polygon has only one slope was the subject of the previous section. We study the case with two slopes in Section \ref{sec:twoslopes}. Finally, we prove Theorem \ref{thm:p_varies}.

\subsection{Proof of Corollary \ref{coro:sensdirect}}

Suppose that the Newton polygon of the equation has two slopes or more. The left and right endpoints of the edge corresponding to the second slope are respectively  $(p^i,\val a_i)$ and $(p^j,\val a_j)$, for some integers $i,j$ with $1 < i < j \leq m$. The associated slope is
$$
\mu_2=\frac{\val a_j - \val a_i}{p^j - p^i}\,.
$$
Since the system is regular singular at $0$, it follows from Theorem \ref{thm:slopes} that the denominator of $\mu_2$ is relatively prime with $p$. Thus, $p$ must divide $\val a_j - \val a_i$. However, since $p>\max\{\val a_j,\val a_i\}$, we must have $\val a_i=\val a_j$, that is $\mu_2 = 0$. 

Suppose that there is a third slope.  Arguing as previously, we obtain $\mu_3=0$, which contradicts the fact that $\mu_3 > \mu_2$.
Thus, the Newton polygon has at most two slopes and, in that case, it follows from the first part of the proof that its second slope is null.

\qed

\subsection{Newton polygons with two slopes}\label{sec:twoslopes}
In this section, we consider a $p$-Mahler equation \eqref{eq:Mahler_at_0} and its associated operator $L$, which satisfy the following assumptions :
\begin{itemize}
    \item[$\rm (i)$] $a_i(z) \in \K[[z]]$ ;
    \item[$\rm (ii)$] $p> \max_i \val(a_i)$ ;
    \item[$\rm (iii)$] the Newton polygon $\mathcal N(L)$ has two slopes, the second one being null.
\end{itemize}

Without loss of generality, we may also assume $\min_i \val a_i=0$, otherwise we may divide each of the power series $a_j$ with $z^{\min_i\val a_i}$. We let $k\in [1,m-1]$ be the integer such that the point $(p^k,\val(a_k))$ is the right endpoint of the edge corresponding to the first slope. Since $\mathcal N(L)$ has two slopes and the second slope is null, we have $\val(a_k)=\val(a_m)=0$ and $\val a_0 \geq 1$.

For every integer $i\in [0,m]$, we write $a_i(z)=\sum_{n=0}^\infty a_{i,n}z^n$ and let
$$
\underline{a_i}(z):=\sum_{n=0}^{\val(a_0)} a_{i,n}z^n\, .
$$ 
Recall that $k\in [1,m-1]$ and $1\leq \val(a_0)<p$. Since $\frac{\val(a_0)}{p^k-1}=-\mu_1$, we have $-\mu_1\in(0,1]$. Moreover $-\mu_1 = 1$ if and only if $ k=1$ and $ \val(a_0) = p-1$.

\begin{prop}\label{prop:a_m_tronques}
   We continue with the assumptions $\rm (i)$ to $\rm (iii)$. 
The equation~\eqref{eq:Mahler_at_0} is regular singular at $0$ if and only if for any exponent $c$ attached to the slope $\mu_2=0$, we have
    \begin{equation}
        \label{eq:condition2slopes}
       \sum_{i=0}^m \underline{a_i}(z)\lambda^i
    \in (\lambda - c)^{m_{c,2}}\K[\lambda,z].
    \end{equation}  
\end{prop}
\begin{proof}

We consider the operator
$$
\underline{L}:=\underline{a_m}(z)\puip^m + \cdots +\underline{a_1}(z)\puip+ \underline{a_0}(z)\,,
$$
so that we may write
$$
L = \underline{L} + z^{\val(a_0)+1}\overline{L} \, ,
$$
for some operator $\overline{L}$ over $\K[[z]]$. With this notation, \eqref{eq:condition2slopes} can be rewritten
\begin{equation}
        \label{eq:condition2slopes-bis}
\underline{L}_\lambda(1)  \in (\lambda - c)^{m_{c,2}}\K[\lambda,z] \,.
\end{equation}
Assume that \eqref{eq:Mahler_at_0} is regular singular at $0$ and let $c$ be attached to $\mu_2$. From Theorem~\ref{thm_RS_si_solutions_tronquees} there exists a truncated solution $f_{c,2}$ associated with $(c,2)$. The least common multiple of the
denominators of the slopes is in this case $d = p^k-1$, the denominator of the first slope. Using Condition \ref{C2}, we can write
$$
f_{c,2}(\lambda,z) = p_0(\lambda) + zp_{1}(\lambda)+ \underbrace{q_1(\lambda)z^{\frac{1}{p^k-1}} + \cdots + q_{p^k}(\lambda) z^{\frac{p^k}{p^k-1}}}_{=:q(\lambda,z)}  \, .
$$
Let $\gamma \in \mathbb Q$ have some denominator $\delta$ relatively prime with $p$. Then, the operator $\underline{L}_\lambda$ sends $z^{\gamma}$ to a linear combination over $\K[\lambda]$ of some $z^\mu$, with $\mu \in \mathbb Q$ some rational numbers whose denominators are also equal to $\delta$. Thus, the only terms which are integer powers of $z$ in the sum 
$$
\underline{L}_\lambda(f_{c,2}(\lambda,z)) = p_0(\lambda)\underline{L}_\lambda(1) + p_1(\lambda)\underline{L}_\lambda(z)+ \underline{L}_\lambda(q(\lambda,z))
$$
 are the one coming from $p_0(\lambda)\underline{L}_\lambda(1) + p_1(\lambda)\underline{L}_\lambda(z)$. Since $f_{c,2}$ satisfies \ref{C1} and since $-\mu_1>0$ it follows that
\begin{equation}\label{eq:val_underline_L}
\val \left(p_0(\lambda)\underline{L}_\lambda(1) + p_1(\lambda)\underline{L}_\lambda(z)\ \, \rm mod\,(\lambda -c)^{s_{c,2}+m_{c,2}}\right) > \val a_0\,.
\end{equation}
Let us prove that
\begin{equation}\label{eq:val_underline_L_0}
\val \left(p_0(\lambda)\underline{L}_\lambda(1)\ \, \rm mod\, (\lambda -c)^{s_{c,2}+m_{c,2}}\right) \geq \val a_0+1>\val a_0\,.
\end{equation}
It is immediate when $p_1(\lambda)=0$. When $p_1\neq 0$, then $-\mu_1=1$ and we must have $k=1$ and $\val a_0 = p-1$. Thus
$$\val(p_1(\lambda)\underline{L}_\lambda(z)) \geq p = \val a_0+1=\val a_0 - \mu_1\,.$$
Thus, \eqref{eq:val_underline_L_0} follows from \eqref{eq:val_underline_L}. 

Since the degree in $z$ of $\underline{L}_\lambda$ is at most $\val a_0$, \eqref{eq:val_underline_L_0} implies that 
\begin{equation}\label{eq:p0_null}
p_0(\lambda)\underline{L}_\lambda(1)  \mod (\lambda -c)^{s_{c,2}+m_{c,2}}=0\,.
\end{equation}
Thus
\begin{equation}
\label{eq:simplif}
    p_0(\lambda)\underline{L}_\lambda(1) = (\lambda-c)^{s_{c,2}+m_{c,2}}Q(\lambda,z)\quad \text{ for some } Q(\lambda,z)\in \K[\lambda,z]\, .
\end{equation}
However, $p_0(\lambda)=\cld(f_{c,2}(\lambda,z))$. Then, it follows form \ref{C3} and the definition of $r_{c,2}(\lambda)$ that ${\rm val}_{\lambda-c}p_0(\lambda)=s_{c,2}$. Dividing by $p_0(\lambda)$ in \eqref{eq:simplif}, we obtain \eqref{eq:condition2slopes-bis}.

\medskip
Let us now prove the converse implication. Assume that \eqref{eq:condition2slopes} holds. Assume first that $-\mu_1<1$. Let us check that $f_{c,2}(\lambda,z)=r_{c,2}(\lambda) \mod (\lambda-c)^{m_{c,2}+s_{c,2}}$ is a truncated solution associated with $(c,2)$ (though it does not depend on $z$):
\begin{itemize}
    \item It follows from \eqref{eq:condition2slopes} 
    and the fact that ${\rm val}_{\lambda - c}(f_{c,2}(\lambda,z)) \geq s_{c,2}$
    that \ref{C1} is satisfied;
    \item Conditions \ref{C2}, \ref{C3} and \ref{C5} are trivially satisfied;
    \item Condition \ref{C4} follows from the fact that $\mu_2=0$.
\end{itemize}
Assume now that $-\mu_1\geq 1$. Since $-\mu_1 \leq 1$ (see the explanation before Proposition \ref{prop:a_m_tronques}), then $-\mu_1=1$ and $\val a_0=p-1$. Write $r_{c,2}(\lambda)=\theta(\lambda)(\lambda - c)^{s_{c,2}}$, with $\theta(\lambda)\in \K[[\lambda-c]]$. By definition of the characteristic polynomial $\chi_1$, we may write $\chi_1(\lambda)=\kappa(\lambda)(\lambda-c)^{s_{c,2}}$ with $\kappa(\lambda) \in \K[\lambda-c]$ such that $\kappa(c)\neq 0$. We claim that
 $$f_{c,2}(\lambda,z) := r_{c,2}(\lambda) - z \theta(\lambda)\kappa(\lambda)^{-1} \sum\limits_{i=0}^m a_{i,p}\lambda^i  \mod (\lambda-c)^{m_{c,2}+s_{c,2}}$$ 
is a truncated solution associated with $(c,2)$. Let us check this.
\begin{itemize}
    \item Since $-\mu_1=1$, $\cld L_\lambda(z)=\chi_1(\lambda)$, by Lemma \ref{lem:chi}. Furthermore, $\val L_\lambda(z)=~p$. Thus, modulo $(\lambda-c)^{m_{c,2}+s_{c,2}}$ and modulo $z^{p+1}=z^{\val a_0-\mu_1+1}$, the following holds
    $$
    L_\lambda(f_{c,2}(\lambda,z))\equiv r_{c,2}(\lambda)\left(L_\lambda(1) - z^p\sum\limits_{i=0}^m a_{i,p}\lambda^i \right)\equiv r_{c,2}(\lambda)\underline{L}_\lambda(1) \equiv 0.
    $$
    Let us prove the first congruence. We have $$L_\lambda(f_{c,2}(\lambda,z)) = r_{c,2}(\lambda)L_{\lambda}(1) -  \big(\theta(\lambda)\kappa(\lambda)^{-1} \sum\limits_{i=0}^m a_{i,p}\lambda^i \mod (\lambda-c)^{m_{c,2}+s_{c,2}}\big) L_\lambda(z)\, .$$
    Moreover, $L_\lambda(z)$ is $\chi_1(\lambda)z^p$ modulo $z^{p+1}$, that is $\kappa(\lambda)(\lambda - c)^{s_{c,2}} z^p$ and we conclude using the definition of $\theta(\lambda)$: $r_{c,2}(\lambda)=\theta(\lambda)(\lambda - c)^{s_{c,2}}$.
   The last congruence above follows from the fact that ${\rm val}_{\lambda - c}(r_{c,2}(\lambda,z))\geq s_{c,2}$ and from \eqref{eq:condition2slopes}. Thus, \ref{C1} is satisfied.
    \item Conditions \ref{C2} to \ref{C5} follow immediately.   
\end{itemize}
Thus, in both situations, $f_{c,2}$ is a truncated solution associated with $(c,2)$. Since the same may be done for any $c$ attached to the slope $\mu_2=0$ and since, from Lemma \ref{lem:first_slope}, there exists a truncated solution associated with $(c,1)$ for any $c$ attached to the slope $\mu_1$, it follows from Theorem \ref{thm_RS_si_solutions_tronquees} that the equation is regular singular at $0$.
\end{proof}

\subsection{Proof of Theorem \ref{thm:p_varies}}
We are now ready to prove Theorem \ref{thm:p_varies}. Without any loss of generality, we may and will suppose that $\min \val a_i = 0$.

Suppose that the equation is regular singular at $0$ when $p$ is equal to some integer $p_1> \max_i \val a_i$. Then, it follows from Corollary \ref{coro:sensdirect} that its Newton polygon has either one slope or two slopes, the second one being null. We study separately both situations.

\subsubsection*{Case 1: The Newton polygon has only one slope}
We consider two subcases.

\subsubsection*{Subcase 1.1: $\val a_0 \geq \val a_m$}
Then $\val a_m=\min_j\val a_j = 0$ and, since the Newton polygon has only one slope, for every integer $i\in [1,m-1]$, 
$$
\frac{\val a_i - \val a_0}{p_1^i - 1} \geq \frac{- \val a_0}{p_1^m - 1}\,.
$$
This can be reformulated as
\begin{equation}\label{eq:val_ai}
\val a_i \geq \val a_0\left(1- \frac{p_1^i - 1}{p_1^m - 1} \right).
\end{equation}
Furthermore, since $p_1 > \val a_0$, we have
$$
\val a_0 \frac{p_1^{i}-1}{p_1^m-1}<\frac{p_1^{i+1}-p_1}{p_1^m-1}<1\, .
$$
Then, it follows from \eqref{eq:val_ai} that $\val a_i > \val a_0 - 1$. Since $\val a_i$ is an integer we eventually obtain $\val a_i \geq \val a_0$. Thus, for any $p\geq 2$, the Newton polygon has only one slope and it follows from Proposition \ref{prop:one_slope} that the equation is regular singular at $0$ for any $p \geq 2$. This proves Subcase 1.1.

\subsubsection*{Subcase 1.2: $\val a_0 < \val a_m$} Since the Newton polygon with the Mahler parameter $p_1$ has only one slope here, then $\val a_0=\min\val a_i = 0$ (see the beginning of this proof) and $\val a_i\geq 1$ for any $i \in [1,m]$. It follows that, for any $p \geq \val a_m$ and any $i \in [1,m]$,
$$
\frac{\val a_i - \val a_0}{p^i - 1} \geq \frac{1}{p^i-1} \geq \frac{1}{p^{m-1}-1} 
\geq \frac{p}{p^m-1} \geq   
\frac{\val a_m}{p^m - 1}=
\frac{\val a_m - \val a_0}{p^m - 1}.
$$
Hence, for any $p \geq \val a_m$, the line from $(1, \val a_0)$ to $(p^m, \val a_m)$ is an edge of the Newton polygon with parameter $p$. Thus, it has only one slope and it follows from Proposition~\ref{prop:one_slope} that the equation is regular singular at $0$. This proves Subcase~1.2.

\medskip

\subsubsection*{Case 2: The Newton polygon has two slopes, the second one being null}
Reasoning as in Subcase 1.1 for each one of the two slopes, we obtain that, for any $p\geq 2$ the Newton polygon has two slopes, the second one being null. Since the equation is regular singular at $0$ when $p$ is equal to $p_1$, it follows from Proposition~\ref{prop:a_m_tronques} that
   $$
    \sum_{i=0}^m \underline{a_i}(z)\lambda^i 
    \in (\lambda - c)^{m_{c,2}}\K[\lambda,z]
    $$
    for all exponent $c$ attached to the slope $\mu_2=0$. 
    However, the above expression do not depend on the integer $p$. Thus, it holds for any $p\geq 2$. Thus, it follows from Proposition \ref{prop:a_m_tronques} that the equation is regular at $0$ for any $p > \max_i \val a_i$. This ends the proof of Theorem \ref{thm:p_varies}.

\end{document}